\documentclass[a4paper,11pt]{article}

\makeatletter
\usepackage{amssymb}
\usepackage{amsthm}
\usepackage{latexsym}
\usepackage{graphicx}
\usepackage{enumerate}
\usepackage{colortbl}
\usepackage{natbib}
\newtheorem{theorem}{Theorem}

\newtheorem{lemma}{Lemma}
\def\argmax{\ensuremath\mathop\mathrm{argmax}}

\def\argmin{\ensuremath\mathop\mathrm{argmin}}

\newcommand{\tabtopsp}[1]{\vbox{\vbox to#1{}\vbox to1em{}}}

\begin{document}
\title{\textbf{\Large Robust estimation of location and concentration parameters for the von Mises--Fisher distribution}}
\author{\scshape{Shogo Kato\thanks{\footnotesize {\it Address for correspondence}: Shogo Kato, The Institute of Statistical Mathematics, 10-3 Midori-cho, Tachikawa, Tokyo 190-8562, Japan. {\tt E-mail:\ skato AT ism.ac.jp}}} \quad {\scshape{and \quad Shinto Eguchi}} \vspace{0.5cm}\\
\textit{The Institute of Statistical Mathematics, Japan}\\}
\date{January 30, 2012}
\maketitle


\begin{abstract}
Robust estimation of location and concentration parameters for the von Mises--Fisher distribution is discussed.
A key reparametrisation is achieved by expressing the two parameters as one vector on the Euclidean space.
With this representation, we first show that maximum likelihood estimator for the von Mises--Fisher distribution is not robust in some situations.
Then we propose two families of robust estimators which can be derived as minimisers of two density power divergences.
The presented families enable us to estimate both location and concentration parameters simultaneously.
Some properties of the estimators are explored.
Simple iterative algorithms are suggested to find the estimates numerically.
A comparison with the existing robust estimators is given as well as discussion on difference and similarity between the two proposed estimators.
A simulation study is made to evaluate finite sample performance of the estimators.
We consider a sea star dataset and discuss the selection of the tuning parameters and outlier detection.
\end{abstract}

\section{Introduction}

Observations which take values on the $p$-dimensional unit sphere arise in various scientific fields.
In meteorology, for example, wind directions measured at a weather station \citep{joh} can be considered two-dimensional spherical or, simply, circular data.
Other examples include directions of magnetic field in a rock sample \citep{ste}, which can be expressed as unit vectors on the three-dimensional sphere. 

For the analysis of spherical data, some probability distributions have been proposed in the literature.
Among them, a model which has played a central role is the von Mises--Fisher distribution which is also called the Langevin distribution.
It has density
$$
f_{\mu,\kappa}(x) =  \frac{ \kappa^{(p-2)/2} }{(2 \pi)^{p/2} I_{(p-2)/2}( \kappa )} \exp \left( \kappa \mu' x \right), \quad x \in S_p,
$$
with respect to surface area on the sphere, where $\mu \in S_p,\ \kappa \geq 0,\ S_p = \{ x \in \mathbb{R}^p \,;\, \|x\|=1 \}$, $y'$ is the transpose of $y$, and $I_q(\cdot)$ denotes the modified Bessel function of the first kind and order $q$ \citep[Equations (8.431) and (8.445)]{gra}.
The parameter $\mu$ controls the centre of rotational symmetry, or the mean direction, of the distribution, while the other parameter $\kappa$ determines the concentration of the model.
The distribution is unimodal and rotationally symmetric about $x=\mu$.
See \citet{wat1983}, \citet{fis1987} and \citet{mar1999} for book treatments of the model.

Although numerous works have been done on robust estimation for models for $\mathbb{R}^p$-valued data in the literature, considerably little attention have been paid to the robust estimation for models for data on a bounded space.
A typical example is a $p$-dimensional sphere which shows some different features from the usual linear space.
Since the unit sphere is a compact set, the gross error sensitivity of the maximum likelihood estimator is bounded.
However, as pointed out, for example, in \citet{wat1983} and discussed later in this paper, there is strong need for the robust estimation for spherical data especially when observations are concentrated toward a certain direction.

There have been some discussion on robust estimation of the parameters for the von Mises--Fisher distribution in the literature.
Robust estimators of the location parameter $\mu$ for the circular, or two-dimensional, case were proposed by \citet[p.28]{mar1972} and \citet{len}.
\citet{fis1985}, \citet{duc1987} and \citet{cha} discussed the estimation of $\mu$ for the general dimensional case.
The estimation of the concentration parameter $\kappa$ was considered by \citet{fis1982}, \citet{duc1990} and \citet{ko}.

As described above, most of these existing works concern robust estimation of either location or concentration parameter for the von Mises--Fisher distribution.
However, comparatively little work has been done to estimate both location and concentration parameters simultaneously.
\citet{len} briefly discussed a numerical algorithm which estimates both parameters for the circular case.
A nonparametric approach is taken in \citet{ago} for estimation for the circular case.
To our knowledge, robust estimation of both parameters for the general dimensional case have never been considered before.

In this paper we propose two families of robust estimators of both location and concentration parameters for the general dimensional von Mises--Fisher distribution.
To achieve this, we first reparametrise the parameters so that they can be expressed as one $\mathbb{R}^p$-valued parameter and then derive the estimators as minimisers of density power divergences developed by \citet{bas} or \cite{jon2001}.
These approaches enable us to estimate both location and concentration parameters simultaneously.
With this parametrisation, some measures of robustness of estimators, such as influence function, are discussed.
To estimate the parameters numerically, we provide simple iterative algorithms.
Some desirable properties such as consistency and asymptotic normality hold for the proposed estimators.
Influence functions and asymptotic covariance matrices are available, and it is shown that they can be expressed in using only the modified Bessel functions of the first kind if a distribution underlying data is a mixture of the von Mises--Fisher distributions.

Subsequent sections are organised as follows.
In Section 1 we discuss maximum likelihood estimation for the von Mises--Fisher distribution and show some problems about the robustness of the estimator.
Also, we briefly consider what is an outlier for spherical data and provide the motivation for our study.
In Sections 2 and 3, we propose two classes of robust estimators of location and concentration parameters and discuss their properties.
A comparison among the two proposed estimators and an existing estimator of \citet{len} is made in Section 4.
In Section 5 a simulation study is given to compare the finite sample performance of the proposed estimators.
In Section 6 a sea star dataset is considered to illustrate how our estimators can be utilised to estimate the parameters and detect outliers.
A prescription for choosing the tuning parameters is given.
Finally, concluding remarks are made in Section 7.

\section{The von Mises--Fisher distribution}

\subsection{Reparametrisation}

Before we embark on discussion on robust estimators, we consider parametrisation of the von Mises--Fisher distribution.
In the most literature on this model, the parameters are represented as a unit vector $\mu$ and a scalar $\kappa$.
Each parameter has clear-cut interpretation;
The parameter $\mu$ controls the mean direction of the model, while $\kappa$ determines the concentration.

In this paper, however, for the sake of discussion on robustness of the estimator, we consider the following reparametrisation:
$$
\xi = \kappa \mu.
$$
Clearly, $\xi$ takes a value in $\mathbb{R}^p$.
It is easy to see that the Euclidean norm of $\xi$, $\|\xi\|$, represents the concentration of the model, while the standardised vector, $\xi/\|\xi\|$, denotes the mean direction.
Then the density of the von Mises--Fisher distribution can be written as
\begin{equation}
f_{\xi}(x) =  \frac{ \| \xi \|^{(p-2)/2} }{(2 \pi)^{p/2} I_{(p-2)/2}( \| \xi \|)} \exp \left( \xi' x \right), \quad x \in S_p; \quad \xi \in \mathbb{R}^p. \label{vm}
\end{equation}
For brevity, write $X \sim \mbox{vM}_p(\xi)$ if an $S_p$-valued random variable $X$ follows a distribution with density (\ref{vm}).
With this convention, it is clearer to evaluate how an outlier influences the estimators of both location and concentration parameters.
For example, the influence function, which is commonly used to discuss robustness, is more interpretable if the parameter is expressed in this manner.
See Sections 2.3, 3.3 and 4.3 for details.
Throughout the paper we denote the density (\ref{vm}) by $f_{\xi}$.

\subsection{Maximum likelihood estimation}
In this subsection we discuss maximum likelihood estimation for the von Mises--Fisher distribution.
Let $X_1,\ldots,X_n$ be random samples from $\mbox{vM}_p(\xi)$.
Then the maximum likelihood estimator of $\xi$ is known to be
\begin{equation}
\hat{\xi}  = A_p^{-1} \Bigl( \mbox{$\frac1n$} \Bigl\| \mbox{$\sum_{j=1}^n X_j$} \Bigl\| \Bigl) \frac{\sum_{j=1}^n X_j}{\|\sum_{j=1}^n X_j\|}, \label{mle}
\end{equation}
where $A_p(x) = I_{p/2}(x)/I_{(p-2)/2}(x), \ x \in [0,\infty) $.
See, for example, \citet[Section 10.3.1]{mar1999} for the derivation of the estimator.
The following hold for $A_p$:
\begin{enumerate}[(i)]
\item $A_p(0) = 0$ and $\lim_{x \rightarrow \infty} A_p(x) = 1$,
\item $A_p(x)$ is strictly increasing with respect to $x$.
\end{enumerate}
See \citet[Appendix A2]{wat1983} for proofs.
From this result, it follows that there exists a unique solution $\hat{\xi}$ which satisfies (\ref{mle}).
These properties are also attractive to solve the inverse function, i.e. $x=A_p^{-1}(y)$, numerically.

Maximum likelihood estimation is associated with minimum divergence estimation based on the Kullback--Leibler divergence.
Let $f_{\xi}$ be density (\ref{vm}) and $G$ the distribution underlying the data having density $g$.
Then the Kullback--Leibler divergence between $f_{\xi}$ and $g$ is defined as
\begin{equation}
d_{KL}(g,f_{\xi}) = \int \log \left(g/f_{\xi} \right) dG(x). \label{kl}
\end{equation}
Here and in many expressions in this paper, we omit the variable of integration.
If we assume that $G$ is the empirical distribution function, i.e., $G=G_n(X_1,\ldots,X_n)$, then the minimiser of the divergence, $\mbox{argmin}_{\xi \in \mathbb{R}^p} d_{KL} (g,f_{\xi})$, is the same as the maximum likelihood estimator (\ref{mle}). 

\subsection{Influence function of the maximum likelihood estimator}

The influence function of the maximum likelihood estimator (\ref{mle}) for the reparametrised von Mises--Fisher distribution (\ref{vm}) is given in the following theorem.
See Appendix for proof.
\begin{theorem} \label{th:if_mle}
The influence function of the maximum likelihood estimator (\ref{mle}) at $G$ is given by
\begin{equation}
\mbox{IF}(G,x) = \{M(\xi)\}^{-1} \left\{ x-A_p (\|\xi\|) \frac{\xi}{\|\xi\|} \right\}, \label{if_mle}
\end{equation}
where
$$
M(\xi)=  \frac{A_p(\| \xi \|)}{\|\xi\|} I + \left\{ 1-A_p^2 (\|\xi\|) - \frac{p}{\|\xi\|} A_p (\|\xi\|) \right\} \frac{\xi \xi'}{\|\xi\|^2}.
$$
\end{theorem}
Note that the above influence function is different from the ones seen in \citet{weh} and \citet{kog}.
Their papers discuss the influence functions of the estimators of location and concentration parameters separately, whereas we summarise these two m.l.e.'s as `one estimator of one parameter' and discuss its influence function.

Given the influence function in Theorem 1, a natural question to address concerns the gross error sensitivity.
Because the unit sphere is a compact set, it is clear that the gross error sensitivity of estimator (\ref{mle}) is bounded.
Nevertheless the following results points out the need for the robust estimation for the model defined on this special manifold.
The proof is straightforward from Theorem \ref{th:if_mle} and omitted.
\begin{theorem}
The following properties hold for maximum likelihood estimator (\ref{mle}):
\begin{enumerate}[(i)]
\item
For any $\xi \in \mathbb{R}^p \setminus \{0\}$, it holds that
$$
\argmax_{x \in S_p} \| \mbox{IF}(G,x) \| = -\frac{\xi}{\|\xi\|} \quad \mbox{and} \quad \argmin_{x \in S_p} \| \mbox{IF}(G,x) \| = \frac{\xi}{\|\xi\|}.
$$
\item Let $\xi/\|\xi\|$ be a fixed vector. 
Then
$$
\lim_{\|\xi\| \rightarrow \infty} \left\{ \sup_{x \in S_p} \| \mbox{IF}(G,x) \| - \inf_{x \in S_p} \| \mbox{IF}(G,x) \|  \right\} = \infty
$$
and
$$
\lim_{\|\xi\| \rightarrow \infty} \left\{ \sup_{x \in S_p} \| \mbox{IF}(G,x) \| \biggl/ \inf_{x \in S_p} \| \mbox{IF}(G,x) \|  \right\} = \infty.
$$
\end{enumerate}
\end{theorem}
This result implies that we need to develop a robust method to estimate $\xi$ when the concentration of the distribution, $\|\xi\|$, is large.

\begin{center}
$***$ Figure 1 about here $***$
\end{center}

Figure 1 plots influence functions (\ref{if_mle}) of maximum likelihood estimators for some selected values of $\xi$.
It seems that the direction of the influence function is close to that of $x$ for small $\|\xi\|$.
The direction is strongly attracted towards $-\xi/\|\xi\|$ when $\|\xi\|$ is large.
The figure also suggests that, for small $\|\xi\|$, the range of the norms of the influence functions is fairly narrow.
The greater the value of $\|\xi\|$, the wider the range of the norm.
As Theorem 2 shows, the norm of the influences functions is maximised if $x=-\xi/\|\xi\|$.
Also, it can confirmed that the norms of the influence functions tend to infinity as $\|\xi\|$ approaches infinity.
This provides strong motivation for robust estimation of the parameter for the von Mises--Fisher distribution.

\subsection{Outliers in directional data} \label{sec:outliers}
Since a unit sphere is a compact set, unlike linear data, it is not clear what is an outlier in directional data.
For example, if a distribution underlying data is the uniform distribution on the sphere, it seems difficult to identify an outlier.
However, if a distribution is highly concentrated, then an outlier can be defined in a similar manner as in linear data.

Here we consider an area where a sample from the von Mises--Fisher density (\ref{vm}) is not likely to be observed.
Let $\alpha \,(\in [0,1])$ be probability which determines the size of the area.
The area, which we denote by $Ar_p$, is defined by the interesection of the unit sphere $S_p$ and the sphere with centre at $-\xi/\|\xi\|$, namely, $ Ar_p = [ x\in S^p; \|x+\xi/\|\xi\|\| < \{ 2(1-\cos \delta) \}^{1/2} ]$.
Here $\delta=\delta(\alpha)\, (\in [0,\pi))$ is the solution of the following equation
$$
\int_{Ar_p} f_{\xi}(x) dx = \alpha. 
$$
The left-hand side of the equation can be simplified to
\begin{eqnarray}
\int_{Ar_p} f_{\xi}(x) dx &=& \int_{\scriptstyle x \in S_p \atop \xi'x/\|\xi\| \leq - \cos \delta} f_{\xi}(x) dx \nonumber \\
&=& \frac{\|\xi\|^{(p-2)/2}}{(2\pi)^{p/2} I_{(p-2)/2}(\|\xi\|)} \int_0^{\pi} \exp (\kappa \cos \theta_1) \sin^{p-2} \theta_1 d \theta_1 \nonumber \\
&& \times \int_0^{2\pi} \int_0^{\pi} \cdots \int_0^{\pi} \sin^{p-3} \theta_{2} \cdots \sin \theta_{p-2} \, d\theta_2 \cdots d \theta_{p-1} \nonumber \\
&=& \frac{(\|\xi\|/2)^{(p-2)/2}}{\Gamma(\frac12) \Gamma \{\frac12 (p-1)\} I_{(p-2)/2}(\|\xi\|)} \int_{-1}^{-\cos \delta} e^{\|\xi\| t} (1-t^2)^{(p-3)/2} dt. \nonumber \\
 \label{area}
\end{eqnarray}
Hence, it follows that $\delta$ can be obtained as the solution of the following integral equation
\begin{equation}
\int^{-\cos \delta}_{-1} e^{\|\xi\| t} (1-t^2)^{(p-3)/2} dt = \pi^{1/2} \alpha \, I_{(p-2)/2} ( \|\xi\| ) \, \Gamma \left\{ \mbox{$\frac12$} (p-1) \right\}  \left( \mbox{$\frac12$} \|\xi\| \right)^{-(p-2)/2}. \label{tail}
\end{equation}
Since the integral in the left-hand side is bounded and strictly increasing with respect $\delta$, it is possible to find the unique solution $\delta$ numerically. 

\begin{center}
$***$ Figure 2 about here $***$
\end{center}

Figure 2(a) demonstrates how the area $Ar_2$ is determined from density (\ref{vm}) and probability $\alpha$.
From this frame, it can be easily understood that the area $Ar_2$ or, equivalently, $(-\delta,\delta)$ increases as $\alpha$ increases. 
Figure 2(b) plots how $\|\xi\|$ influences the size of area (\ref{area}).
As this frame clearly shows, the size of the area, which is monotonically increasing with respect to $\delta$, increases as $\|\xi\|$ increases.
It can be confirmed that, for any $p$, as $\|\xi\|$ tends to infinity, $\delta$ approaches $\pi$, meaning that a sample is likely to be observed only in a neighbour of $x=\xi/\|\xi\|$.
Therefore we conclude that robust estimation for the von Mises--Fisher is necessary especially when the parameter $\|\xi\|$ is large.
This statement is also supported from discussion given in Figure 1 in the previous subsection.

\section{Minimum divergence estimator of \citet{bas}}

In this section we propose a family of estimators of the parameter for the von Mises--Fisher distribution.
Our estimator can be derived as a minimiser of the divergence proposed by \citet{bas}.
An iterative algorithm is presented to estimate the parameter numerically.
The influence function and asymptotic distribution of the estimator are considered.

\subsection{The divergence of \citet{bas}}

Let $f_{\theta}$ be a parametric density and $g$ a density underlying the data.
\citet{bas} define the density power divergence between $g$ and $f_{\theta}$ to be
\begin{equation}
d_{\beta} (g,f_{\theta}) = \int \left\{ \frac{1}{\beta (1+\beta)} g^{1+\beta} - \frac{1}{\beta} g f_{\theta}^{\beta} + \frac{1}{1+\beta} f_{\theta}^{1+\beta} \right\} dx , \quad \beta >0, \label{beta_div}
\end{equation}
$$
d_{0} (g,f_{\theta}) = \lim_{\beta \rightarrow 0} d_{\beta} (g,f_{\theta}) = \int g \log (g/f_{\theta}) dx.
$$
This divergence is called the $\beta$-divergence as seen in \citet{min} and \citet{fuj2006}.

The divergence between the von Mises--Fisher density and a density underlying the data is given in the following theorem.
The proof is given in Appendix.
\begin{theorem} \label{thm:beta_div}
Let $f_{\theta}$ in (\ref{beta_div}) be the von Mises--Fisher $\mbox{vM}_p(\xi)$ density.
Then the Basu \textit{et al.}\ divergence is
\begin{eqnarray}
d_{\beta} (g,f_{\xi}) &=& \frac{1}{\beta (1+\beta)} \int g^{1+\beta} dx - \frac{1}{\beta} \, \left\{ \frac{\| \xi \|^{(p-2)/2}}{(2\pi)^{p/2} I_{(p-2)/2} (\| \xi \|)} \right\}^{\beta}  \int  \exp (\beta \xi' x) g(x) dx   \nonumber \\
&& + \frac{\| \xi \|^{(p-2)\beta/2}}{ (2\pi)^{p \beta/2} (1+\beta)^{p/2}} \frac{I_{(p-2)/2} \{ (1+\beta) \| \xi \|\}}{I^{1+\beta}_{(p-2)/2} (\| \xi \|)}  , \quad \beta>0. \label{beta_vm}
\end{eqnarray}
If $\beta$ equals 0, then $d_{0}(g,f_{\xi}) = d_{KL} (g,f_{\xi}),$ where $d_{KL}(g,f_{\xi})$ is as in (\ref{kl}).
\end{theorem}

As a density underlying the data, consider the following mixture model
\begin{equation}
g(x)=(1-\varepsilon) f_{\xi}(x)+\varepsilon f_{\eta}(x), \label{mix}
\end{equation}
where $0 \leq \varepsilon \leq 1$ and $f_{\zeta}$ denotes the von Mises--Fisher $\mbox{vM}_p(\zeta)$ density.
In this case the divergence can be expressed as
\begin{eqnarray*}
d_{\beta}(g,f_{\xi}) &=& \int \{ (1-\varepsilon) f_{\xi} +\varepsilon f_{\eta} \}^{1+\beta} dx +  \frac{1}{\beta} \left\{ \frac{\| \xi \|^{(p-2)/2}}{(2\pi)^{p/2} I_{(p-2)/2} (\| \xi \|)} \right\}^{\beta} \\
&& \times \left( \varepsilon \left[ \frac{I_{(p-2)/2} \{(1+\beta) \|\xi\|\}}{ (1+\beta)^{(p-2)/2} I_{(p-2)/2} (\|\xi \|)} - \frac{ \|\eta\|^{(p-2)/2}  }{ \| \beta \xi + \eta \|^{(p-2)/2} } \frac{ I_{(p-2)/2} (\| \beta \xi +\eta \|) }{ I_{(p-2)/2} (\|\eta \|)} \right] \right. \\
&& \left. - \frac{I_{(p-2)/2} \{(1+\beta) \|\xi\|\}}{ (1+\beta)^{p/2} I_{(p-2)/2} (\|\xi \|)} \right).
\end{eqnarray*}
Note that the second and third terms of the divergence can be expressed by using only the modified Bessel functions of the first kind.
In general, the first term should be evaluated numerically.
If $\beta$ is an integer, then the first term is of the form
\begin{eqnarray*}
\int \{ (1-\varepsilon) f_{\xi} +\varepsilon f_{\eta} \}^{1+\beta} dx &=& \sum_{k=0}^{1+\beta} {1+\beta \choose k}  \frac{ (1-\varepsilon)^k \varepsilon^{1+\beta-k} ( \| \xi\|^{k}\, \|\eta\|^{1+\beta-k} )^{(p-2)/2} }{ (2\pi)^{\beta p/2} \|k \xi + (1+\beta-k) \eta \|^{(p-2)/2}} \\
 && \times \frac{I_{(p-2)/2} \{ \|k \xi + (1+\beta-k) \eta \| \}}{ I^k_{(p-2)/2} (\|\xi\|) I_{(p-2)/2}^{1+\beta-k} (\|\eta\|) }.
\end{eqnarray*}
It is remarked here that, in this case, the first term also does not involve any special functions other than the modified Bessel functions of the first kind.

\subsection{Estimating equation} \label{sec:ee}

The estimating equation derived from the \citet{bas} divergence is known to be
\begin{equation}
\int \tilde{\psi}_{\beta}(x,\xi) dG(x) = 0, \label{beta}
\end{equation}
where
$$
\tilde{\psi}_{\beta}(x,\xi) = f_{\xi}^{\beta} u_{\xi} - \int f_{\xi}^{1+\beta} \, u_{\xi} \, dy \quad \mbox{and} \quad u_{\xi} = \frac{\partial}{\partial \xi} \log f_{\xi} = x - A_p(\|\xi\|) \frac{\xi}{\|\xi\|}.
$$
Following the convention in \citet{jon2001}, we call the solution of this equation function the type 1 estimator.
From the general theory, it immediately follows that the estimator is consistent for $\xi$.
\begin{theorem} \label{th:psi_beta}
The function $\tilde{\psi}_{\beta}(x,\xi)$ can be expressed as
\begin{eqnarray*}
\tilde{\psi}_{\beta} (x,\xi) &=&  C^{\beta} \Biggl( \left\{ x -A_p (\| \xi \|) \frac{\xi}{\|\xi\|} \right\} \exp (\beta \xi' x) \\ 
&& - \frac{I_{(p-2)/2} \{(1+\beta) \|\xi\| \} }{(1+\beta)^{(p-2)/2} I_{(p-2)/2} (\|\xi\|) } \left[ A_p \{(1+\beta)\|\xi\|\} - A_p (\| \xi\| ) \right] \frac{\xi}{\|\xi\|} \Biggr),
\end{eqnarray*}
where $C$ is the normalising constant of the von Mises--Fisher $\mbox{vM}_p (\xi)$ density, i.e., $C= \|\xi\|^{(p-2)/2}$ $/ \{ (2 \pi)^{p/2} I_{(p-2)/2}(\|\xi\|) \}$.
\end{theorem}
See Appendix for the proof.
From this form, it immediately follows that Fisher consistency holds for the estimator.
For simplicity, we redefine the $\psi$-function as
\begin{equation}
\psi_{\beta}(x,\xi) = C^{-\beta} \tilde{\psi}_{\beta} (x,\xi). \label{psi_beta2}
\end{equation}
Then Equation (\ref{beta}) holds for $\psi_{\beta}$ replaced by $\tilde{\psi}_{\beta}$.
Since $C$ does not depend on $x$, it is clear that $M$-estimation based upon $\psi_{\beta}$ is essentially the same as the one based on $\tilde{\psi}_{\beta}$.

Substituting $G$ for an empirical distribution $G_n(X_1,\ldots,X_n)$ in Equation (\ref{beta}) in which $\tilde{\psi}_{\beta}$ is replaced by $\psi_{\beta}$, we have
\begin{eqnarray*}
\lefteqn{ n^{-1} \sum_{j=1}^n \left\{ x_j - A_p (\|\xi\|) \frac{\xi}{\|\xi\|} \right\} \exp (\beta \xi' x_j) } \\
&& - \frac{I_{(p-2)/2} \{(1+\beta) \|\xi\|\}}{ (1+\beta)^{(p-2)/2} I_{(p-2)/2} (\|\xi \|)} \, \left[ A_p \left\{ (1+\beta) \|\xi\| \right\} -A_p(\|\xi\|) \right] \frac{\xi}{\|\xi\|} = 0.
\end{eqnarray*}
Thus the estimator which minimises divergence (\ref{beta_div}) satisfies the following relationship:
\begin{equation}
\hat{\xi} = A_p^{-1} \left\{ \frac{ \| \sum_j w_{j,\beta} (\hat{\xi}) x_j - n \tilde{D}_{\hat{\xi}} \, \hat{\xi} \| }{\sum_j w_{j,\beta} (\hat{\xi}) } \right\} \frac{\sum_j w_{j,\beta} (\hat{\xi}) x_j }{ \| \sum_j w_{j,\beta} (\hat{\xi})  x_j  \| }, \label{esti_beta}
\end{equation}
where $w_{j,\beta} (\xi) = \exp (\beta \xi' x_j)$ and $ \tilde{D}_{\xi} = I_{(p-2)/2} \{(1+\beta) \| \xi \| \} [ A_p \{(1+\beta)\| \xi \|\} - A_p (\| \xi \|) ] / \{  (1+\beta)^{(p-2)/2} I_{(p-2)/2} (\| \xi \|) \, \|\xi \| \} $.
Note that, since $A_p(x) (\equiv y)$ is strictly increasing with respect to $x$, there exists a unique solution $x$ satisfying $x=A_p^{-1}(y)$. 

Then an algorithm induced from the above relationship is suggested as follows. \vspace{0.3cm}

\noindent {\bf Algorithm for the type 1 estimate}

\begin{enumerate}[\hspace{1cm}]
\item[\textit{Step} 1.] Take an initial value $\xi_0$.
\item[\textit{Step} 2.] Compute $\xi_1,\ldots, \xi_N$ as follows until the estimate $\xi_{N}$ remains virtually unchanged from the previous estimate $\xi_{N-1}$,
$$
\xi_{t+1} = A_p^{-1} \left\{ \frac{ \| \sum_j w_{j,\beta}(\xi_t)\, x_j - n \tilde{D}_{\xi_t} \xi_t \| }{\sum_j w_{j,\beta} (\xi_t) } \right\} \frac{\sum_j w_{j,\beta} (\xi_t) \, x_j }{ \| \sum_j w_{j,\beta} (\xi_t) \,  x_j  \| }.
$$
\item[\textit{Step} 3.] Record $\xi_N$ as an estimate of $\xi$.
\end{enumerate}
Our simulation study implies that the above algorithm converges if an initial value is set properly and $\beta$ is not too large.
As an initial value, the maximum likelihood estimator (\ref{mle}) may be one promising choice.

The tuning parameter $\beta$ can be estimated by using the cross-validation \citep[Section 7.10.1]{has}.
We will discuss more details of the selection of $\beta$ in Section 7.

\subsection{Influence function}

In this subsection we consider the influence function of the type 1 estimator and compare this estimator with m.l.e.\ in terms of the influence function.

\begin{theorem} \label{thm:if_beta}
The influence function of the type 1 estimator at $G$ is given by
\begin{equation}
\mbox{IF} (x,G) = \left\{ M_{\beta}(\xi,G) \right\}^{-1} \psi_{\beta}(x,\xi), \label{if_beta}
\end{equation}
where
\begin{eqnarray}
M_{\beta}(\xi,G) &=&  - \int \exp (\beta \xi' x) \left[ \beta x x' - \beta A_p(\|\xi\|) \frac{\xi x'}{\|\xi\|} -  \frac{A_p (\|\xi\|)}{\|\xi\|} I  \right. \nonumber \\
&& \left. - \left\{ 1- \frac{p}{\|\xi\|} A_p(\|\xi\|) -A_p^2 (\|\xi\|) \right\} \frac{\xi \xi'}{\|\xi\|^2} \right] dG(x) \nonumber \\
&& - \left\{ (1+\beta)^{(p-2)/2} I_{(p-2)/2}(\|\xi\|) \right\}^{-1} \nonumber \\
&& \times \Biggl\{ \|\xi\|^{-1} \left[ A_p\{(1+\beta)\|\xi\|\}- A_p (\|\xi\|) \right] I_{(p-2)/2} \{(1+\beta) \|\xi\|\}  I  \nonumber \\
&& + \Biggl( (1+\beta) I_{p/2}\{(1+\beta)\|\xi\|\} \left[ A_p\{(1+\beta)\|\xi\|\} - A_p (\|\xi\|) \right] \nonumber \\
&& + I_{(p-2)/2}\{(1+\beta)\|\xi\|\} \Biggl[ \beta - (1+\beta) A_p^2 \{(1+\beta)\|\xi\|\} \nonumber \\
&&  - \frac{p}{\|\xi\|} A_p \{(1+\beta)\|\xi\|\} - A_p(\|\xi\|) A_p \{(1+\beta)\|\xi\|\} + 2 A^2_p(\|\xi\|) \nonumber \\
&& \left. \left. \left. + \frac{p}{\|\xi\|} A_p(\|\xi\|) \right] \right) \frac{\xi \xi'}{\|\xi\|^2} \right\}, \label{av}
\end{eqnarray}
and $\psi_{\beta}(x,\xi)$ is as in (\ref{psi_beta2}).
\end{theorem}

\begin{proof}
The influence function (\ref{if_beta}) can be obtained in a similar approach as in Theorem 1.
Some calculations to obtain $M_{\beta}$ can be done by using Theorem \ref{th:psi_beta}, Equations (\ref{app1}) and (\ref{app2}), and Equation (8.431.1) of \citet{gra}.
\end{proof}

Here we consider the mixture model (\ref{mix}) as a distribution of $G$.
Then the integral part of function $M_{\beta}(\xi,G)$ in the influence function is given by
\begin{eqnarray*}
\lefteqn{ - \int \exp (\beta \xi' x) \left[ \beta x x' - \beta A_p(\|\xi\|) \frac{\xi x'}{\|\xi\|} - \frac{A_p (\|\xi\|)}{\|\xi\|} I  \right. } \nonumber \\
&& \left. - \left\{ 1- \frac{p}{\|\xi\|} A_p(\|\xi\|) -A_p^2 (\|\xi\|) \right\} \frac{\xi \xi'}{\|\xi\|^2} \right] dG(x) \\
&=& \frac{1-\varepsilon}{(1+\beta)^{(p-2)/2} I_{(p-2)/2}(\|\xi\|)} \\
&& \times \left( \left[ A_p(\|\xi\|) I_{(p-2)/2}\{(1+\beta)\|\xi\|\} - \frac{\beta I_{p/2} \{(1+\beta)\|\xi\|\}}{(1+\beta) \|\xi\|} \right] I \right. \\
&& - \left[ \left\{ \beta -1 + \frac{p A_p(\|\xi\|)}{\|\xi\|} +A^2_p(\|\xi\|) \right\} I_{(p-2)/2}\{(1+\beta)\|\xi\|\}  \right. \\
&& \left. + \left\{ \frac{p}{(1+\beta)\|\xi\|} + A_p(\|\xi\|) \right\} I_{p/2}\{(1+\beta)\|\xi\|\} \right] \frac{\xi \xi'}{\|\xi\|^2} \biggr) \\
&& - \frac{\|\eta\|^{(p-2)/2}}{\|\zeta_{1,\beta} \|^{(p-2)/2}} \frac{\varepsilon}{I_{(p-2)/2}(\| \eta \|)} \left[ \left\{ \beta I_{p/2} (\|\zeta_{1,\beta}\|) - I_{(p-2)/2} (\|\zeta_{1,\beta}\|) \frac{A_p(\|\xi\|)}{\|\xi\|} \right\} I \right. \\
&& - \left\{ I_{(p-2)/2}(\|\zeta_{1,\beta}\|) -\frac{p}{\|\zeta_{1,\beta}\|} I_{p/2}(\|\zeta_{1,\beta}\|) \right\} \frac{\zeta_{1,\beta} \zeta'_{1,\beta}}{\|\zeta_{1,\beta}\|^2} \\
&& + A_p (\|\xi\|) I_{p/2}(\|\zeta_{1,\beta}\|) \frac{\xi \zeta_{1,\beta}'}{\|\xi\| \|\zeta_{1,\beta}\| }\\
&& \left. + I_{(p-2)/2}(\|\zeta_{1,\beta}\|) \left\{ 1-A_p^2(\|\xi\|) -\frac{p A_p(\|\xi\|)}{\|\xi\|} \right\} \frac{\xi \xi'}{\|\xi\|^2} \right],
\end{eqnarray*}
where $\zeta_{j,\alpha} = j \alpha \xi + \eta $.
Note that, in this case, the influence functions do not involve any special functions other than the modified Bessel functions of the first kind and orders $(p-2)/2$ and $p/2$.

\begin{center}
$***$ Figure 3 about here $***$
\end{center}

Figure 3 displays influence functions (\ref{if_beta}) of the type 1 estimator at the two-dimensional von Mises--Fisher model for some selected values of $\beta$.
From four frames of the figure and Figure 1(d), it seems that the norms of influence functions are not large for moderately large $\beta$.
In particular, it seems that, for $\beta=0.25$,
$
\| \mbox{IF}(G_{VM},-\xi/\|\xi\|) \| - \| \mbox{IF}(G_{VM},\xi/\|\xi\|) \| 
$
and
$
\| \mbox{IF}(G_{VM},-\xi/\|\xi\|) \| / \| \mbox{IF}(G_{VM},\xi/\|\xi\|) \|
$
take smaller values than those for the maximum likelihood estimator, where $G_{VM}$ is the distribution function of $\mbox{vM}_2\{(2.37,0)'\}$.
This result implies that the type 1 estimator is more robust than the maximum likelihood estimator.

\subsection{Asymptotic normality} \label{an_beta}
The asymptotic normality of the estimator can be shown from the $M$-estimation theory.
Let $X_1,\ldots,X_n$ be random samples from the von Mises-Fisher $\mbox{vM}_p(\xi)$ distribution.
Suppose that $\hat{\xi}$ is the type 1 estimator of $\xi$.
Then
$$
n^{1/2} (\hat{\xi} - \xi ) \stackrel{\rm d}{\longrightarrow} N(0,V_{\beta}) \quad \mbox{as} \quad n \rightarrow \infty,
$$
where
$$
V_{\beta} = M_{\beta}(\psi_{\beta},G)^{-1} Q_{\beta}(\psi_{\beta},G) \{M_{\beta}(\psi_{\beta},G)'\}^{-1},
$$
$$ Q_{\beta}(\psi_{\beta},G) = \int \psi_{\beta}(x,\xi) \{ \psi_{\beta} (x,\xi) \}' dG(x),
$$
and $\psi_{\beta}(x,\xi)$ and $M_{\beta}(\psi_{\beta},G)$ are defined as in (\ref{psi_beta2}) and (\ref{av}), respectively.
In particular, if $G$ is the distribution function of the mixture model (\ref{mix}), then $Q_{\beta}(\psi_{\beta},G)$ is given by
\begin{eqnarray*}
Q_{\beta} (\psi_{\beta},G) &=& \frac{1-\varepsilon}{I_{(p-2)/2} (\|\xi\|)} \Biggl\{ \frac{I_{p/2} \{(1+2\beta) \|\xi\|\} }{ (1+2 \beta)^{p/2} \|\xi\| } I \\
 && + \Biggl( \frac{A_p(\|\xi\|)}{(1+2\beta)^{(p-2)/2}} \left[ A_p(\|\xi\|) I_{(p-2)/2} \{(1+2\beta)\|\xi\|\} - 2 I_{p/2} \{(1+2\beta)\|\xi\|\} \right] \\
 && - \frac{1-A_p(\|\xi\|)}{(1+\beta)^{p-2}} \frac{I^2_{p/2} \{(1+\beta)\|\xi\|\}}{I_{(p-2)/2}(\|\xi\|)} \left[ A_p\{(1+\beta)\|\xi\|\} -A_p(\|\xi\|) \right] \Biggr) \frac{\xi \xi'}{\|\xi\|^2} \Biggr\} \\
 && +  \frac{\varepsilon}{I_{(p-2)/2}(\|\eta\|)} \Biggl( \frac{\|\eta\|^{(p-2)/2}}{\|\zeta_{2,\beta}\|^{(p-2)/2}} \biggl[ \frac{I_{p/2} (\| \zeta_{2,\beta} \|)}{ \| \zeta_{2,\beta} \| } I \\
&& + \left\{ I_{(p-2)/2}(\|\zeta_{2,\beta}\|) - \frac{p}{\|\zeta_{2,\beta}\|} I_{p/2}(\|\zeta_{2,\beta}\|) \right\} \frac{ \zeta_{2,\beta} \zeta_{2,\beta}' }{\| \zeta_{2,\beta} \|^2} \\
&& -A_p(\|\xi\|) I_{p/2} (\|\zeta_{2,\beta}\|) \frac{\zeta_{2,\beta} \xi' + \xi \zeta_{2,\beta}'}{\|\zeta_{2,\beta}\| \|\xi\|} + A_p^2 (\|\xi\|) I_{(p-2)/2} (\|\zeta_{2,\beta}\|) \frac{\xi \xi'}{\|\xi\|^2} \biggr] \\
&& - \frac{\|\eta\|^{(p-2)/2}}{\|\zeta_{1,\beta}\|^{(p-2)/2}} \frac{I_{(p-2)/2} \{(1+\beta)\|\xi\|\}}{(1+\beta)^{(p-2)/2} I_{(p-2)/2} (\|\xi\|)} \left[ A_p\{(1+\beta)\|\xi\|\} - A_p (\|\xi \|) \right] \\
&& \times \left\{ I_{p/2}(\|\zeta_{1,\beta}\|) \frac{\zeta_{1,\beta} \xi' + \xi \zeta_{1,\beta}'}{\|\zeta_{1,\beta}\| \|\xi\|} + 2 A_p(\|\xi\|) I_{(p-2)/2} (\|\zeta_{1,\beta}\|) \frac{\xi \xi'}{\|\xi\|^2} \right\} \Biggr) \\
&& + \frac{I^2_{(p-2)/2} \{(1+\beta) \|\xi\|\} }{(1+\beta)^{p-2} I^2_{(p-2)/2} (\|\xi\|) } \left[ A_p \{(1+\beta)\|\xi\|\} - A_p (\| \xi\| ) \right]^2 \frac{\xi \xi'}{\|\xi\|^2}.
\end{eqnarray*}
Remark that the asymptotic covariance matrix can be expressed in a form of the modified Bessel functions of the first kind and orders $(p-2)/2$ and $p/2$.

\section{Minimum divergence estimator of \citet{jon2001}}

\subsection{The divergence of \citet{jon2001}}
Next we consider another divergence which is also based on density power.
It is defined as
\begin{equation}
d_{\gamma} (g,f_{\theta}) = \frac{1}{\gamma (1+\gamma)} \log \int g^{1+\gamma} dx - \frac{1}{\gamma} \log \int g f_{\theta}^{\gamma} dx + \frac{1}{1+\gamma} \log \int f_{\theta}^{1+\gamma} dx, \quad \gamma>0, \label{gamma}
\end{equation}
$$
d_0(g,f_{\theta}) = \lim_{\gamma \rightarrow 0} d_{\gamma} (g,f_{\theta}) = \int g \log ( g / f_\theta) dx.
$$
This divergence was briefly considered by \citet[Equation (2.9)]{jon2001} as a special case of a general family of divergences.
\citet{fuj2008} investigated detailed properties of the divergence with emphasis on the case in which the underlying distribution contains heavy contamination.

The divergence between the von Mises--Fisher density and a density underlying the data can be calculated as follows.
The procedure to derive this divergence is similar to that to obtain the \citet{bas} divergence given in Theorem \ref{thm:beta_div}, and therefore the proof is omitted.
\begin{theorem}
Let $f_{\xi}$ be the von Mises--Fisher $\mbox{vM}_p(\xi)$ density.
Then divergence (\ref{gamma}) between $f_{\xi}$ and an arbitrary density $g$ is given by
\begin{eqnarray}
d_{\gamma}(g,f_{\xi}) &=& \frac{1}{\gamma (1+\gamma)} \log \int g^{1+\gamma} dx - \frac{1}{\gamma} \log \int \exp (\gamma \xi' x) g(x) dx \nonumber \\
 && + \frac{1}{1+\gamma} \log \left[  \frac{(2 \pi)^{p/2} I_{(p-2)/2} \{\| (1+\gamma) \xi \|\}}{  \| (1+\gamma) \xi\|^{(p-2)/2}} \right]. \label{gamma_vm}
\end{eqnarray}
\end{theorem}
Note that the expression for this divergence is slightly simpler than that for the \citet{bas} divergence.

As a special case of the underlying density in the Jones \textit{et al.}\ divergence, consider again the mixture of the two von Mises--Fisher densities (\ref{mix}).
Then the divergence can be expressed as
\begin{eqnarray*}
d_{\gamma}(g,f_{\xi}) &=& \frac{1}{\gamma (1+\gamma)} \log \int \{(1-\varepsilon) f_{\xi} + \varepsilon f_{\eta} \}^{1+\gamma} dx \\
 && -\frac{1}{\gamma} \log \left[ (1-\varepsilon) \frac{I_{(p-2)/2} \{(1+\gamma) \|\xi\|\}}{ (1+\gamma)^{(p-2)/2} I_{(p-2)/2} (\|\xi \|)} + \varepsilon \frac{ \| \xi \|^{(p-2)/2} }{ I_{(p-2)/2} (\|\xi\|)} \frac{I_{(p-2)/2} (\| \gamma \xi +\eta \|)}{ \| \gamma \xi +\eta \|^{(p-2)/2} } \right] \\
 && - \frac{1}{1+\gamma} \log \left[  \frac{(2 \pi)^{p/2} I_{(p-2)/2} \{\| (1+\gamma) \xi \|\}}{ \|(1+\gamma) \xi\|^{(p-2)/2}} \right].
\end{eqnarray*}
From discussion in \citet{fuj2008}, one can ignore the contamination $f_{\eta}$ if the second term in the second logarithm is sufficiently small.
This condition is satisfied, for example, when $\| \xi \|$ is large and $ \| \gamma \xi + \eta \| \simeq 0$ holds.

\subsection{Estimating equation}
The estimating equation of the Jones \textit{et al.}\ divergence is
\begin{equation}
\frac{\int f_{\xi}^{\gamma} \, u_{\xi} \, g\, dx }{\int f_{\xi}^{\gamma}\, g \, dx} - \frac{\int f_{\xi}^{1+\gamma} \, u_{\xi} \, dx}{ \int f_{\xi}^{1+\gamma} dx } = 0, \label{estimating}
\end{equation}
where $u_{\xi}$ is defined as in Section 2.3.
Or equivalently,
$$
 \int \tilde{\psi}_{\gamma}(x,\xi) dG(x) = 0, \quad 
\mbox{where} \quad
\tilde{\psi}_{\gamma}(x,\xi) =  f_{\xi}^{\gamma} \left( u_{\xi} - \frac{\int f_{\xi}^{1+\gamma} \, u_{\xi} \, dy}{ \int f_{\xi}^{1+\gamma} dy } \right).
$$
It is remarked here that this equation has been discussed by \citet{win} although the divergence which the estimating equation is based on was not considered there.
As in \citet{jon2001}, the estimator derived from this estimation equation is called the type 0 estimator in the paper.
\begin{theorem} \label{th:psi_gamma}
The function $\tilde{\psi}_{\gamma}(x,\xi)$ is given by
$$
\tilde{\psi}_{\gamma}(x,\xi) = C^{\gamma} \exp (\gamma \xi' x) \left[ x -A_p\{(1+\gamma) \|\xi\|\} \frac{\xi}{\|\xi\|} \right],
$$
where $C$ is as in Theorem \ref{th:psi_beta}.
\end{theorem}
The proof is similar to that for Theorem \ref{th:psi_beta} and omitted.
In a similar manner as in Section \ref{sec:ee}, we redefine the $\psi$-function as $\psi_{\gamma}(x,\xi)=\tilde{\psi}_{\gamma}(x,\xi)$ and consider the $M$-estimation based on $\psi_{\gamma}$.

On substituting $G$ for an empirical distribution $G_n(X_1,\ldots,X_n)$ in (\ref{estimating}), it follows that
$$
\sum_{j=1}^n \exp (\gamma \xi' x_j) x_j \biggl/ \sum_{j=1}^n \exp (\gamma \xi' x_j)  - A_p\{ (1+\gamma) \|\xi\| \} \frac{\xi}{\|\xi\|} = 0.
$$
Therefore it can be seen that the minimum divergence estimator satisfies the following equation
\begin{equation}
\hat{\xi} = \frac{1}{1+\gamma} A_p^{-1} \left\{ \frac{ \| \sum_j w_{j,\gamma} (\hat{\xi}) x_j \| }{ \sum_j w_{j,\gamma} (\hat{\xi}) } \right\} \frac{\sum_j w_{j,\beta}(\hat{\xi}) x_j}{ \| \sum_j w_{j,\beta} (\hat{\xi}) x_j \|}. \label{esti_gamma}
\end{equation}

Given estimating Equation (\ref{esti_gamma}), the following algorithm is naturally induced.\vspace{0.3cm}

\noindent {\bf Algorithm for the type 0 estimate}

\begin{enumerate}[\hspace{1cm}]
\item[\textit{Step} 1.] Take an initial value $\xi_0$.
\item[\textit{Step} 2.] Compute $\xi_1,\ldots,\xi_N$ as follows until the estimate $\xi_N$ remains virtually unchanged from the previous estimate $\xi_{N-1}$,
$$
\xi_{t+1} = \frac{1}{1+\gamma} A_p^{-1} \left\{ \frac{\| \sum_j w_{j,\gamma} (\xi_t) x_j \|}{\sum_j w_{j,\gamma} (\xi_t) } \right\} \frac{\sum_j w_{j,\gamma} (\xi_t) x_j}{\| \sum_j w_{j,\gamma} (\xi_t) x_j \|}.
$$
\item[\textit{Step} 3.] Record $\xi_N$ as an estimate of $\xi$.
\end{enumerate}
This algorithm can also be derived from an iterative algorithm of \citet[Section 4]{fuj2008}, and from their discussion, the monotonicity of the algorithm follows:
$$
d_{\gamma} (\overline{g},f_{\xi_0}) \geq d_{\gamma} (\overline{g},f_{\xi_{1}}) \geq \cdots \geq d_{\gamma} (\overline{g},f_u),
$$
where $\overline{g}$ denotes the empirical density.

As for a prescription for choosing the tuning parameter $\gamma$, cross-validation is available.
See Section 7 for details.

\subsection{Influence function}
The influence function of the type 0 estimator (\ref{estimating}) is provided in the following theorem.
The basic process to obtain the divergence is similar to that in Theorem \ref{thm:if_beta} and omitted.

\begin{theorem} \label{thm:if_gamma}
The influence function of the type 0 estimator at $G$ is given by
\begin{equation}
\mbox{IF}(x;\xi,G) = \{M_{\gamma}(\xi,G)\}^{-1} \, \psi_{\gamma}(x,\xi), \label{if_gamma}
\end{equation}
where
\begin{eqnarray*}
M_{\gamma}(\xi,G) &=& - \int \exp (\gamma \xi' x) \left\{ \gamma \left[ x x' -  A_p \{ (1+\gamma) \| \xi \| \} \frac{ \xi x'}{\|\xi\|} \right] - \frac{ A_p \{ (1+\gamma) \| \xi \| \}}{\|\xi\|} I \right. \\
&& - \left. \left( (1+\gamma) \left[ 1 - A_p^2 \left\{ (1+\gamma) \| \xi \| \right\} \right] + \frac{p A_p \{(1+\gamma)\|\xi\|\}}{\|\xi\|} \right) \frac{ \xi \xi' }{\| \xi \|^2} \right\} dG(x).
\end{eqnarray*}
\end{theorem}
If $G$ is the distribution function of the mixture model with density (\ref{mix}), then the function $M_{\gamma}(\xi,G)$ in Theorem \ref{thm:if_gamma} can be expressed as
\begin{eqnarray*}
M_{\gamma}(\xi,G) &=&  \frac{1-\varepsilon}{ (1+\gamma)^{(p-2)/2} I_{(p-2)/2} (\|\xi\|)} \left[ \frac{1}{1+\gamma} \frac{ I_{p/2}\{(1+\gamma) \|\xi\|\}}{\|\xi\|} I \right. \\
&& + \left\{ \left[ \frac{p}{(1+\gamma) \|\xi\|} + A_p \{(1+\gamma)\|\xi\|\} \right] \gamma \, I_{p/2} \{(1+\gamma)\|\xi\|\} + I_{(p-2)/2}\{(1+\gamma)\|\xi\|\} \right. \\
&& \times \left. \left. \left( (1+\gamma) [1-A_p^2\{(1+\gamma)\|\xi\|\}] - \frac{p A_p \{(1+\gamma) \|\xi\|\}}{\|\xi\|} -\gamma \right) \right\} \frac{\xi \xi'}{\|\xi\|^2} \right] \\
&& - \frac{\|\eta\|^{(p-2)/2}}{\|\zeta_{1,\gamma}\|^{(p-2)/2}} \frac{\varepsilon}{I_{(p-2)/2} (\|\eta\|)} \left\{ \left[ \frac{\gamma I_{p/2} (\|\zeta_{1,\gamma}\|) }{\|\zeta_{1,\gamma}\|} - I_{(p-2)/2}(\|\zeta_{1,\gamma}\|) \frac{A_p\{(1+\gamma)\|\xi\|\}}{\|\xi\|} \right] I \right. \\
&& + \gamma \left\{ I_{(p-2)/2} (\|\zeta_{1,\gamma}\|) - \frac{p}{\|\zeta_{1,\gamma}\|} I_{p/2} (\|\zeta_{1,\gamma}\|) \right\} \frac{\zeta_{1,\gamma} \zeta_{1,\gamma}'}{\|\zeta_{1,\gamma}\|^2} \\
&& - \gamma \, A_p\{(1+\gamma)\|\xi\|\} I_{p/2} (\|\zeta_{1,\gamma}\|) \frac{\xi \zeta_{1,\gamma}'}{\|\xi\| \|\zeta_{1,\gamma}\|} \\
&& \left. - I_{(p-2)/2} (\|\zeta_{1,\gamma}\|) \left( (1+\gamma) [ 1-A_p^2 \{(1+\gamma)\|\xi\|\}] - \frac{p A_p\{(1+\gamma) \|\xi\|\}}{\|\xi\|} \right) \frac{\xi \xi'}{\|\xi\|^2} \right\}.
\end{eqnarray*}
\vspace{0cm}

\begin{center}
$***$ Figure 4 about here $***$
\end{center}

Figure 4 plots the influence functions (\ref{if_beta}) of the type 0 estimator for four selected values of $\gamma$.
Note that the values of the tuning parameters in the fours frames of this figure correspond to those in Figure 3.
Comparing these two figures, it seems that the influence functions of both estimators take quite similar values when both tuning parameters are the same.

\subsection{Asymptotic normality}
In a similar way as in Section \ref{an_beta}, one can show the asymptotic normality of the estimator.
Since $\psi_{\gamma}(x,\xi)$ and $M_{\gamma}(\psi_{\gamma},G)$ have already been given in Theorem \ref{thm:if_gamma}, the function $Q_{\gamma}(\psi_{\gamma},G)$ and the asymptotic covariance matrix $V_{\gamma}$ can be calculated in a straightforward manner.
In particular, if $G$ is the distribution function of the mixture model with density (\ref{mix}), then $Q_{\gamma}(\psi_{\gamma},G)$ is given by
\begin{eqnarray*}
Q_{\gamma}(\psi_{\gamma},G) &=& \frac{1-\varepsilon}{(1+2\gamma)^{(p-2)/2} I_{(p-2)/2} (\|\xi\|)} \left( \frac{I_{p/2}\{(1+2\gamma)\|\xi\|\}}{(1+2\gamma) \|\xi\| } I + \biggl[ I_{(p-2)/2}\{(1+2\gamma)\|\xi\|\} \right. \\
&& - \frac{p I_{p/2}\{(1+2\gamma)\|\xi\|\}}{(1+2\gamma) \|\xi\|} -2 A_p\{(1+\gamma)\|\xi\|\} I_{p/2}\{(1+2\gamma)\|\xi\|\} \\
&& \left. + A_p^2 \{(1+\gamma)\|\xi\|\} I_{(p-2)/2}\{(1+2\gamma)\|\xi\|\} \biggr] \frac{\xi \xi'}{\|\xi\|^2} \right) \\
&& + \frac{\|\eta\|^{(p-2)/2}}{\| \zeta_{2,\gamma} \|^{(p-2)/2}} \frac{ \varepsilon  }{ I_{(p-2)/2} (\|\eta\|) } \left[ \frac{I_{p/2}(\|\zeta_{2,\gamma}\|)}{ \| \zeta_{2,\gamma} \| } I + \biggl\{ I_{(p-2)/2}(\|\zeta_{2,\gamma}\|) \right. \\
&& - \frac{p I_{p/2} (\|\zeta_{2,\gamma}\|)}{\|\zeta_{2,\gamma}\|} \biggl\} \frac{\zeta_{2,\gamma} \zeta_{2,\gamma}'}{\|\zeta_{2,\gamma}\|^2} - A_p\{(1+\gamma)\|\xi\|\} I_{p/2}(\|\zeta_{2,\gamma}\|) \frac{ \zeta_{2,\gamma} \xi'+\xi \zeta_{2,\gamma}'}{\| \xi \| \| \zeta_{2,\gamma} \|} \\
&& \left. + A_p^2\{(1+\gamma) \| \xi \| \} I_{(p-2)/2}(\|\zeta_{2,\gamma}\|) \frac{\xi \xi'}{\|\xi\|^2} \right].
\end{eqnarray*}
Again, the asymptotic covariance matrix does not involve any special functions other than the modified Bessel functions of the first kind and orders $(p-2)/2$ and $p/2$.

\section{Comparison among the robust estimators}

\subsection{Comparison between types 0 and 1 estimators}

In this subsection we compare the two proposed estimators of the parameters of the von Mises--Fisher distribution.
A detailed comparison between the two estimators for the general family of distributions has been given in \citet{jon2001}.
Here we consider some properties of the estimators which are special for the von Mises--Fisher distribution.

The numerical algorithms for both types of estimates presented in Sections 3.2 and 4.2 have some similarities and differences.
A similarity is that both algorithms are expressed in relatively simple forms and require to calculate the inverse of the function $A_p$, i.e., the ratio of the modified Bessel functions.
However, as discussed in Section 2.2, it seems fairly easy to calculate $A_p^{-1}$ numerically since it is bounded and strictly increasing.
A slight difference is that the algorithm for the type 0 estimator appears to be slightly simpler than that for the type 1 estimator as it does not involve subtraction in the argument of $A_p^{-1}$ in Step 2.
Another special feature of the type 0 estimator is the monotonicity of the algorithm as shown in Section 4.2.

Next we discuss the influence functions and asymptotic covariance matrices of both estimators, which can be expressed by using the modified Bessel functions of the first kind and orders $(p-2)/2$ and $p/2$.
Their expression for the type 0 estimator is simpler although both require the calculations of the aforementioned functions.
A comparison between Figures 3 and 4 suggests that the behaviour of the influence functions of both estimators looks similar, at least, for the practical choices of the tuning parameters.
When the tuning parameters of both estimators are large, e.g., $\beta=\gamma=0.75$, a simulation study, which will be given in the next section, implies that the influence functions of each estimator behave in a different manner.
However, in most of the realistic cases in which the tuning parameters are moderately small, the performance of the estimators and their influence functions seems quite similar.
As for the asymptotic covariance matrices, as discussed in \citet[Section 3.3]{jon2001}, the large-sample variances matrices of both estimators show relatively small loss of efficiency when the tuning parameters are equal and small.
We will provide further comparison of these estimators through a simulation study and an example in the later sections.

\subsection{Similarities to and differences from \citeauthor{len}'s (1981) estimator}

\citet{len} proposed a robust estimator of a location parameter of the two-dimensional von Mises--Fisher distribution and briefly considered an algorithm to estimate both location and concentration parameters.
The estimator is defined as follows.
Assume that $\theta_1,\ldots,\theta_n$ are random samples from the von Mises distribution vM${}_2(\kappa \cos \mu,\kappa \sin \mu)$.
Define
$$
\overline{C}_w = \frac{\sum_j w_j \cos \theta_j}{\sum_j w_j}, \quad
\overline{S}_w = \frac{\sum_j w_j \sin \theta_j}{\sum_j w_j}, \quad
\overline{R}_w = \left\{ \overline{C}^2_w + \overline{S}^2_w \right\}^{1/2},
$$
$$
w_j = \frac{ \psi \{ t(\theta_j -\mu;\kappa) \}}{t(\theta_j-\mu;\kappa)}, \quad t(\phi;\kappa) = \pm \{ 2 \kappa (1-\cos \phi) \},
$$
\begin{equation}
\psi_H (t)=\left\{
\begin{array}{ll}
t, & |t| \leq c\\
c \ \mbox{sign} (t), & |t|>c
\end{array}
\right.
,
\quad
\psi_A (t)=\left\{
\begin{array}{ll}
c \, \sin (t/c), & |t| \leq c \pi\\
0, & |t|>c \pi
\end{array}
\right.
. \label{weight}
\end{equation}
with $``+"$ or $``-"$ chosen according to $\phi$ (mod $2\pi) \in (\mbox{or}\notin)\ [0,\pi)$.
Then the estimator is defined as solutions of the following estimating equations:
$$
\sum_{j=1}^n w_j \sin (\theta_j - \hat{\mu})= 0 \quad \mbox{and} \quad \hat{\kappa} = A_2^{-1} (\overline{R}_w).
$$

This estimator is somewhat associated with the types 0 and 1 estimators discussed in the paper.
All of these three estimators are related in the sense that the parameters are estimated by introducing some weight functions in the estimating equations.
Also, all estimators can be obtained numerically through fairly simple algorithms.

However our two families of estimators are different from Lenth's one.
One obvious distinction is that our estimators discuss the general dimensional case of the von Mises--Fisher distribution, while the Lenth estimator, as it stands, can be used only for the two-dimensional case.
In addition there are some other differences between the Lenth estimator and ours even for the two-dimensional case.
As seen in Equations (\ref{esti_beta}) and (\ref{esti_gamma}), our estimators adopt the power of the densities as weight functions, whereas \citet{len} used the weight functions (\ref{weight}) proposed by \citet{hub} or \citet{and}.

This distinction makes a difference in discussing Fisher consistency of the estimators.
As shown in Sections 3.2 and 4.2 of the paper, our estimators are Fisher consistent.
On the other hand, as shown below, Fisher consistency does not hold for the Lenth estimator.
To prove this, we first show the following general result, which helps us evaluate theoretical first cosine moment for the Lenth estimator.
The proof is given in Appendix.
\begin{lemma}
Let $f$ be a probability density function on the circle $[-\pi,\pi)$.
Assume that $w$ is a function on $[-\pi,\pi)$ which satisfies the following properties:
\begin{enumerate}
\item $w$ is symmetric about 0, i.e., $w(\theta)=w(-\theta)$ for any $\theta \in [-\pi,\pi)$.
\item If $\cos \theta_1 > \cos \theta_2$, then $w(\theta_1) \geq w(\theta_2)$.
\item $0 < \int_{-\pi}^{\pi} w(\theta) f(\theta) d\theta < \infty$.
\end{enumerate}
Then
$$
\frac{\int_{-\pi}^{\pi} w(\theta) \cos \theta f(\theta) d\theta }{\int_{-\pi}^{-\pi} w(\theta) f(\theta) d\theta} \geq \int_{-\pi}^{\pi} \cos \theta f(\theta) d\theta.
$$
The equality holds if and only if $w(\theta)=c.$
\end{lemma}

Using Lemma 1, we immediately obtain the following result.
See Appendix for the proof.
\begin{theorem}
Assume $\psi$ is not a constant function.
Then Lenth's estimator is not Fisher consistent.
\end{theorem}
However, we should note that Lenth's estimator of the location $\mu$, which is the main focus of the paper, is Fisher consistent if the concentration $\kappa$ is known, and the estimator can be useful in that situation.

\section{Simulation study}

In this section a simulation study is carried out to compare the finite sample performance of the two proposed estimators.

We consider the performance of the estimators in the following two cases: (i) random samples of some selected sizes are gathered from the von Mises--Fisher distributions (without contamination), and (ii) 100 random samples are generated from the contaminated von Mises--Fisher distributions with some selected contamination ratios.
The case (i) is to discuss how much efficiency of the estimators is lost when random samples of small sizes do not include any outliers.
The more attention will be paid to the case (ii) as robustness is the main theme of the paper.
We do not discuss \citeauthor{len}'s (1981) estimator here since, as shown in the previous section, the estimator for the concentration parameter can be biased due to the fact that Fisher consistency does not hold.

First consider the case (i) in which finite samples are generated from the von Mises--Fisher distribution without contamination.
Random samples of sizes $n=10,\ 20,\ 30,\ 50$ and $100$ from the von Mises--Fisher distributions $\mbox{vM}_p(\xi)$ with $p=2$ and $\xi=(2.37,0)'$ and $p=3$ and $\xi=(3.99,0,0)'$ are generated.
For each combination of $n$ and $\xi$, 2000 simulation samples are gathered.
We discuss the performance of the estimators in terms of the mean squared error.
The estimate of the mean squared error is given by $\sum_{j=1}^{2000} \|\hat{\xi}_j -\xi \|^2/2000$, where the $\hat{\xi}_j$'s $(j=1,\ldots,2000)$ are the estimates of $\xi$ for $j$th simulation sample.

\begin{center}
$***$ Table 1 about here $***$
\end{center}

Table 1 shows the estimates of the relative mean squared errors of the types 0 and 1 estimators for some selected values of the tuning parameters.
A comparison of these two estimators suggests that, when the tuning parameters of these estimators are equal, the estimates of the relative mean squared errors generally take similar values.
An exception is a case in which the sample size is small and the tuning parameters of the estimators are large.
In this case the type 1 estimator generally outperforms the type 0 estimator although both estimators do not seem satisfactory.
Except for this special case, however, it might be appreciated that only a little efficiency is lost for these robust estimators.
The table suggests that the relative mean squared error diminishes as the tuning parameter decreases.
Also it is noted that, for large sample sizes, the relative mean squared errors of both estimators are almost equal to one regardless of the values of tuning parameters, confirming consistency of the estimators.

Next we consider the case (ii) in order to discuss the robustness of these two families of estimators is discussed.
Two families of distributions are chosen as contaminations, namely, the uniform and von Mises--Fisher distributions.
The uniform distribution or the von Mises--Fisher distribution with small concentration $\|\xi\|$ is often used as a contamination as seen in \citet{duc1987} and \citet{cha}.
It seems less common to assume the von Mises--Fisher distribution with fairly large concentration parameter as a contamination, but this model also appears to be a reasonable choice if we choose its parameter such that most observations from the model lie on an area where samples from the von Mises--Fisher of interest are not likely to be observed.

First consider the uniform contamination.
Generate 100 random samples from a mixture of the von Mises--Fisher and uniform distributions having the form $(1-\varepsilon) \, \mbox{vM}_p (\xi) + \varepsilon \, U_p$ for some selected values of $\varepsilon,\ p$ and $\xi$.
Then we calculate the estimates of the relative mean squared error in a similar way as in the previous simulation.

\begin{center}
$***$ Table 2 about here $***$
\end{center}

Table 2 displays the estimates of the relative mean squared errors of types 0 and 1 estimators with respect to maximum likelihood estimator for some selected values of tuning parameters.
It seems from the table that, when the concetration parameter or, equivalently, $A_p(\|\xi\|)$ is small, the relative mean squared errors are close to one.
This can be mathematically validated from the fact that the von Mises--Fisher distribution approaches the uniform distribution as $\|\xi\|$ tends to 0.
On the other hand, when $A_p(\|\xi\|)$ is large, then the robust estimators outperform the maximum likelihood estimator.
In particular, when $A_p(\|\xi\|)$ is close to 1, the relative mean squared errors of the proposed estimators take much smaller values than one.
It is also noted that the tuning parameters which minimise the relative mean squared errors increase as the contamination ratio $\varepsilon$ increases.

Second, we discuss the robustness of the estimators when the true distribution follows a mixture of two von Mises--Fisher distributions.
This time, generate 100 samples from a mixture of the two von Mises--Fisher distributions having the form $(1-\varepsilon)\, \mbox{vM}_p (\xi) + \varepsilon \, \mbox{vM}_p (\zeta)$ with some selected values of $\varepsilon,\ p,\ \xi$ and $\zeta$.

\begin{center}
$***$ Table 3 about here $***$
\end{center} 

Estimates of the relative mean squared errors of the types 0 and 1 estimators with respect to maximum likelihood estimator for some selected contamination ratios and tuning parameters for a mixture of two von Mises--Fisher distributions are given in Table 3.
Note that the contaminating distribution is assumed to follow the von Mises--Fisher, not the uniform distribution, so that almost all observations lie in the area where observations from the distribution of interest are not likely to be observed.
This table implies that the two estimators outperform the maximum likelihood estimator if the tuning parameters are chosen correctly.
In particular, if $\varepsilon$ is large, both of the proposed estimators show much better results than the maximum likelihood estimator.
It seems from the table that both estimators behave quite similarly, especially for small values of tuning parameters.
Since the contaminating distribution is concentrated toward the opposite direction of $\xi$, the tuning parameters minimising the relative mean squared errors are greater than those given in Table 2 for the fixed values of $\varepsilon$.

\section{Example}

To illstrate how our methods can be utilised to real data, we consider a dataset of directions of sea stars \citep[Example 4.20]{fis1993}.
The dataset consists of the resulant directions of 22 Sardinian sea stars 11 days after being displaced from their natural habitat.

\begin{center}
$***$ Figure 5 about here $***$
\end{center}

Figure 5(a) plots measurements of resultant directions of sea stars.
Since the dataset shows symmetry and unimodality, it seems reasonable to fit the von Mises distribution to this dataset.
However, as this frame suggests, there are two observations which can be considered possible outliers.
Of these two samples, one at 2.57 seems to be an apparent outlier on the assumption that the observations follow a von--Mises distribution, while the other one at 5.20 appears to be much more difficult to judge.
We fit the von Mises distribution based on the maximised likelihood and types 1 and 0 divergences and discuss how these results can be utilised for detecting outliers.
To select the tuning parameters of the types 1 and 0 estimators, we use the three-fold cross-validation \citep[Section 7.10.1]{has} implemented as follows.
First, divide the dataset $D$ into three subsets $D_1,D_2$ and $D_3$.
Define
\begin{equation}
\mbox{CV}(\hat{f}_{\xi},\alpha) = \frac{1}{N} \sum_{j=1}^N L \left\{ y_j,\hat{f}_{\xi}^{-\tau(j)}(x_j,\alpha)\right\}, \label{cv}
\end{equation}
where $\hat{f}_{\xi}^{-\iota}(x,\alpha)$ is the estimated density with a tuning parameter $\alpha$ based on a subset of the data $D \setminus D_{\iota}$, $N$ is the sample size of the dataset, and $\tau(k)$ is an index function defined as $ \tau(k) = l$ for $x_k \in D_l$.
Here we define the loss functions $L$ for the types 1 and 0 estimators as $L \left\{y,\hat{f}_{\xi}^{-\iota}(x,\alpha)\right\} = d_{0.6} (\overline{g}_y,\hat{f}_{\xi}^{-\iota})$ where $d_{0.6}$ are the Basu \textit{et al.}\ divergence (\ref{beta_vm}) with $\beta=0.6$ and Jones \textit{et al.}\ divergence (\ref{gamma_vm}) with $\gamma=0.6$, respectively, in which $\overline{g}_y$ is a probability function of a point distribution with singularity at $Y=y$.
Then the estimate of the tuning parameter is given by a minimiser of CV$(\hat{f}_{\xi},\alpha)$.
Figure 5(b) and (c) exhibit the values of CV$(\hat{f}_{\xi},\alpha)$ for the types 1 and 0 estimators, respectively, for $\alpha=h/100\ (h=1,\ldots,100)$.
The curves of the frames show somewhat similar behaviours when the tuning parameters take values between 0 and 0.45, while they look different if the tuning parameters are greater 0.45.
These frames suggest that, for the Basu \textit{et al.}\ divergence, the values of CV are more stable than the Jones \textit{et al.}\ divergence for the tuning parameter greater than 0.45.

\begin{center}
$***$ Table 4 about here $***$
\end{center}

Table 4 shows the estimates of the parameters and tuning parameters for the maximum likelihood and types 1 and 0 estimators.
The maximum likelihood estimators are obtained for some subsets of the data which exclude no samples, one at 2.57 and ones at 2.57 and 5.20.
A comparison among the maximum likelihood estimates suggests that these possible outliers do not influence the estimate of the location parameter $\mu$ significantly.
On the other hand, the estimate of the concentration parameter $\kappa$ seems to be influenced by the possible outliers.
Both types 1 and 0 estimates are similar to the maximum likelihood estimate for the dataset excluding the one sample, implying that the dataset includes only one outlier at 2.57 actually.
This conclusion coincides with the one given by \citet[Example 4.20]{fis1993} who derived the same consequence from his outlier test for discordancy for von Mises data.
Figure 5(d) and (e) display Q--Q plots for the data excluding one outlier for types 1 and 0 estimators, respectively, where quantiles of the robust estimators (horizontal axis) and of the empirical distribution (vertical axis) are plotted.
This figure shows that the estimated model provides a satisfactory fit to the dataset.

\section{Discussion}

As pointed out in \citet{wat1983} and some other references, it is known that maximum likelihood estimator of the parameter for the von Mises--Fisher distribution is not robust.
In particular, as discussed in Section 2.3, a robust estimator is required especially when observations are concentrated toward a certain direction.
\citet{len} briefly considered an algorithm to estimate both location and concentration parameters simultaneously.
However, as discussed in Section 5.2, Lenth's estimator of the parameters can be used only for the circular case of the distribution and is not Fisher consistent.
In this paper we provided two families of robust estimators which enable us to estimate both location and concentration parameters simultaneously for the general case of the von Mises--Fisher distribution.
Both estimators can be derived as the minimisers of divergences proposed by \citet{bas} and \citet{jon2001}.
It follows from the general theory that some properties, including consistency and asymptotic normality, hold for the estimators.
In addition it was shown our estimators have some special features.
For example, the presented estimators can be obtained through fairly simple algorithms numerically.
Also, the influence functions and asymptotic covariance matrices of both estimators can be expressed using the modified Bessel functions of the first kind.
Some simulations suggest that the performance of both estimators is quite satisfactory and, in particular, the proposed estimators greatly outperform the maximum likelihood estimator if the distributions underlying data are concentrated toward a certain direction.
Possible future works include robust estimation of parameters for the extended families of distributions such as the ones proposed by \citet{ken} and \citet{jon2005}.
In particular robust methods for distributions with weak symmetry properties would be desired.

\section*{Acknowledgement}
Financial support for the research of the first author was provided, in part, by the Ministry of Education, Culture, Sport, Science and Technology in Japan under a Grant-in-Aid for Young Scientists (B) (20740065).

\appendix

\section{Proof of Theorem 1}
From $M$-estimation theory \citep[Section 4.2c]{ham}, the influence function of the maximum likelihood estimator (\ref{if_mle}) is of the form
$$
IF(G,x) = \left\{ M(\xi) \right\}^{-1} \psi(x,\xi),
$$
where
$$
\psi(x,\xi) = \frac{\partial}{\partial \xi} \log f_{\xi} \quad \mbox{and} \quad M(\xi) = -\int \left. \frac{\partial}{\partial \zeta} \psi (x,\zeta) \right|_{\xi} dF(x).
$$
After some algebra, it follows that
\begin{eqnarray*}
\psi(x,\xi) &=& \frac{\partial}{\partial \xi} \log f_{\xi} \\
 &=& \frac{\partial}{\partial \xi} \left\{ \frac{p-2}{2} \log \|\xi\| - \log I_{(p-2)/2}(\|\xi\|) + \xi' x \right\} \\
 &=& \frac{p-2}{2} \frac{\xi}{\|\xi\|^2} - \frac{ (p-2) I_{(p-2)/2}(\|\xi\|) / (2 \|\xi\|) + I_{p/2} (\|\xi\|)}{I_{(p-2)/2}(\|\xi\|)} \frac{\xi}{\|\xi\|} + x \\
 &=& x - A_p(\|\xi\|) \frac{\xi}{\|\xi\|},
\end{eqnarray*}
where the third equality derives from the following formula \citep[9.6.26]{abr}:
\begin{equation}
\left. \frac{d}{ds} I_p(s) \right|_t = \frac{p}{t} I_p(t) + I_{p+1}(t). \label{app1}
\end{equation}
Using this result, $M(\xi)$ can be calculated as
\begin{eqnarray*}
M(\xi) &=& - \int \left. \frac{\partial}{\partial \zeta'} \psi (x,\zeta) \right|_{\xi} dF(x) \\
 &=& -\int \left. \frac{\partial}{\partial \zeta'} \left\{ x - A_p(\|\zeta\|) \frac{\zeta}{\|\zeta\|} \right\} \right|_{\xi} dF(x) \\
 &=& \left. \frac{\partial}{\partial t} A_p(t) \right|_{\|\xi\|} \frac{\xi \xi'}{\|\xi\|^2} + A_p (\|\xi\|) \left\{ \frac{1}{\|\xi\|} \left( I - \frac{\xi \xi'}{\|\xi\|^2} \right) \right\} \\
 &=& \left\{ 1 - A_p^2 (\|\xi\|) - \frac{p-1}{\|\xi\|} A_p(\|\xi\|) \right\} \frac{\xi \xi'}{\|\xi\|^2} + A_p (\|\xi\|) \left\{ \frac{1}{\|\xi\|} \left( I - \frac{\xi \xi'}{\|\xi\|^2} \right) \right\} \\
 &=& \frac{A_p(\|\xi\|)}{\|\xi\|} I + \left\{ 1 - A_p^2 (\|\xi\|) - \frac{p}{\|\xi\|} A_p(\|\xi\|) \right\} \frac{\xi \xi'}{\|\xi\|^2},
\end{eqnarray*}
where the fourth equality holds due to the following formula (\citet[Appendix 1, (A.14)]{mar1999}; \citet[p.289]{jam})
\begin{equation}
\left. \frac{d}{ds} A_p(s) \right|_t = 1 - A_p^2(t) - \frac{p-1}{t} A_p(t). \label{app2}
\end{equation}
Thus we obtain Theorem 1. \quad $\Box$

\section{Proof of Theorem \ref{thm:beta_div}}
It is clear that $d_{0}(g,f_{\xi}) = d_{KL} (g,f_{\xi})$.
We consider a case $\beta>0$.
\begin{eqnarray*}
d_{\beta} (g,f_{\xi}) &=& \int \left\{ \frac{1}{\beta (1+\beta)} g^{1+\beta} - \frac{1}{\beta} g f_{\xi}^{\beta} + \frac{1}{1+\beta} f_{\xi}^{1+\beta} \right\} dx \\
 &=& \frac{1}{\beta (1+\beta)} \int g^{1+\beta} dx - \frac{1}{\beta} \, \left\{ \frac{\| \xi \|^{(p-2)/2}}{(2\pi)^{p/2} I_{(p-2)/2} (\| \xi \|)} \right\}^{\beta}  \int  \exp (\beta \xi' x) g(x) dx \\
 && + \frac{1}{1+\beta} \left\{ \frac{\| \xi \|^{(p-2)/2}}{(2\pi)^{p/2} I_{(p-2)/2} (\| \xi \|)} \right\}^{1+\beta} \int \exp \{(1+\beta) \xi'x \} dx
\end{eqnarray*}
Using the fact that $\int f_{(1+\beta)\xi}\, dx=1$, the integral in the third term of the equation can be expressed as
\begin{eqnarray*}
\int \frac{1}{1+\beta} f_{\xi}^{1+\beta} \, dx &=& \frac{1}{1+\beta} \left\{ \frac{\| \xi \|^{(p-2)/2}}{(2\pi)^{p/2} I_{(p-2)/2} (\| \xi \|)} \right\}^{1+\beta} \frac{2 \pi^{p/2} I_{(p-2)/2} \{(1+\beta) \|\xi\|\}}{\{ (1+\beta) \|\xi\|/2\}^{(p-2)/2}} \\
 &=& \frac{\| \xi \|^{(p-2)\beta/2}}{ (2\pi)^{p \beta/2} (1+\beta)^{p/2}} \frac{I_{(p-2)/2} \{ (1+\beta) \| \xi \|\}}{I^{1+\beta}_{(p-2)/2} (\| \xi \|)}. \hspace{2cm} \Box
\end{eqnarray*}

\section{Proof of Theorem \ref{th:psi_beta}}
It is easy to obtain the first term of $\tilde{\psi}_{\beta}$.
We consider the second term of $\tilde{\psi}_{\beta}$, namely, $- \int f_\xi^{1+\beta} u_{\xi} dy$.
To be more specific, it can be expressed as
\begin{eqnarray}
\int f_\xi^{1+\beta} u_{\xi} dy &=& C^{1+\beta} \int \exp \{ (1+\beta) \xi' y \} \left\{ x - A_p(\|\xi\|) \frac{\xi}{\|\xi\|} \right\} dy \nonumber \\
 &=& C^{1+\beta} \left[ \int y \exp \{ (1+\beta) \xi' y \} dy - A_p (\|\xi\|) \frac{\xi}{\|\xi\|} \cdot \int \exp \{ (1+\beta) \xi' y \} dy \right]. \nonumber \\
&& \label{integration}
\end{eqnarray}
The first integration can be calculated by using the fact that, if $X \sim \mbox{vM}_p (\zeta)$, then $E(X)=A_p(\zeta)\, \zeta/\|\zeta\|$. (See, for example, \citet[p.169]{mar1999}).
From this, it immediately follows that
$$
\int y \exp \{(1+\beta) \xi' y \} dy = \frac{(2\pi)^{p/2} I_{p/2} \{ (1+\beta)\|\xi\| \}}{\{(1+\beta)\|\xi\|\}^{(p-2)/2} } \frac{\xi}{\|\xi\|}.
$$
The process to calculate the second integration of (\ref{integration}) is essentially the same as that to obtain the normalising constant of the von Mises--Fisher density $\mbox{vM}_p\{(1+\beta)\xi\}$.
Thus we obtain Theorem \ref{th:psi_beta}. \quad $\Box$

\section{Proof of Lemma 1}
For convenience, write
$$
T_w = \frac{\int_{-\pi}^{\pi} w(\theta) \cos \theta f(\theta) d\theta }{\int_{-\pi}^{-\pi} w(\theta) f(\theta) d\theta} \quad \mbox{and} \quad T = \int_{-\pi}^{\pi} \cos \theta f(\theta) d\theta.
$$
Then $T_w$ can be expressed as
$
T_w = \int w'(\theta) \cos \theta f (\theta) d\theta,
$
where $w'(\theta)= w(\theta)$ $/ \int_{-\pi}^{\pi} w(u) f(u)du$.
With this convention, it holds that
\begin{eqnarray}
T_w &=& \int_{-\pi}^{\pi} w'(\theta) \cos \theta f(\theta) d\theta \nonumber \\
&=& T + \int_{-\pi}^{\pi} \{ w'(\theta) -1\} \cos \theta f(\theta) d\theta. \label{second}
\end{eqnarray}
The second term of (\ref{second}) can be decomposed into two terms as
\begin{eqnarray}
\int_{-\pi}^{\pi} \{ w'(\theta) -1\} \cos \theta f(\theta) d\theta 
&=& \int_{w'(\theta) \geq 1} + \int_{w'(\theta) < 1}. \label{two_parts}
\end{eqnarray}
Due to Properties 1-3 of $w(\theta)$, it is easy to see that there exists a constant $d \in [-1,1)$ such that $\{ \theta \in [-\pi,\pi) \, | \, w'(\theta) \geq 1 \} = \{ \theta \in [\pi,\pi) \, | \, \cos \theta \geq d \}$.
Then the following inequality holds for (\ref{two_parts}):
\begin{eqnarray*}
\int_{w'(\theta) \geq 1} + \int_{w'(\theta) < 1} &=& \int_{\cos \theta \geq d } \{ w'(\theta) -1 \} \cos \theta f(\theta) d\theta + \int_{\cos \theta <d} \{ w'(\theta) -1 \} \cos \theta f(\theta) d\theta \\
&\geq& d \left[ \int_{\cos \theta \geq d } \{ w'(\theta) -1 \} f(\theta) d\theta + \int_{\cos \theta <d} \{ w'(\theta) -1 \} f(\theta) d\theta \right] \\
&=& d \left\{ \int_{-\pi}^{\pi} w'(\theta) f(\theta) d\theta - \int_{-\pi}^{\pi} f(\theta) d\theta \right\} \\
&=& 0.
\end{eqnarray*}
Thus we obtain $T_w \geq T$.
Since $f(\theta)>0$ for some subset $A$ which is not a null set, it can be seen that the equality holds if and only if $w'(\theta)=1$.  \quad $\Box$

\section{Proof of Theorem 7}
We show that $\hat{\kappa}$ is not a Fisher consistent estimator.
Without loss of generality, assume $\mu=0$.
Then it is easy to see that $\int w(\theta) \sin \theta f_{VM} (\theta) d \theta=0$, where $f_{VM}$ is the von Mises $\mbox{vM}_2 (\kappa,0)$ density since the integrand is an odd function.
Then, from Lemma 1, we immediately obtain
\begin{eqnarray*}
\frac{ \int w(\theta) \cos \theta f_{VM}(\theta) d\theta }{\int w(\theta) f_{VM}(\theta) d\theta} &>& \int \cos \theta f_{VM}(\theta) d\theta \\
&=& A_2(\kappa).
\end{eqnarray*}
Therefore $R_w = M_w > A_2(\kappa)$. \quad $\Box$





\bibliography{Kato_Eguchi}

\begin{thebibliography}{32}
\expandafter\ifx\csname natexlab\endcsname\relax\def\natexlab#1{#1}\fi
\providecommand{\bibinfo}[2]{#2}
\ifx\xfnm\relax \def\xfnm[#1]{\unskip,\space#1}\fi
\bibitem[{Abramowitz and Stegun(1970)}]{abr}
\bibinfo{author}{M.~Abramowitz}, \bibinfo{author}{I.A. Stegun},
  \bibinfo{title}{{Handbook of Mathematical Functions}},
  \bibinfo{publisher}{Dover Press}, \bibinfo{address}{New York},
  \bibinfo{year}{1970}.
\bibitem[{Agostinelli(2007)}]{ago}
\bibinfo{author}{C.~Agostinelli}, \bibinfo{title}{Robust estimation for
  circular data}, \bibinfo{journal}{Computational Statistics and Data Analysis}
  \bibinfo{volume}{51} (\bibinfo{year}{2007}) \bibinfo{pages}{5867--5875}.
\bibitem[{Andrews(1974)}]{and}
\bibinfo{author}{D.F. Andrews}, \bibinfo{title}{A robust method for multiple
  linear regression}, \bibinfo{journal}{Technometrics} \bibinfo{volume}{16}
  (\bibinfo{year}{1974}) \bibinfo{pages}{523--531}.
\bibitem[{Basu et~al.(1998)Basu, Harris, Hjort and Jones}]{bas}
\bibinfo{author}{A.~Basu}, \bibinfo{author}{I.R. Harris}, \bibinfo{author}{N.L.
  Hjort}, \bibinfo{author}{M.C. Jones}, \bibinfo{title}{Robust and efficient
  estimation by minimising a density power divergence},
  \bibinfo{journal}{Biometrika} \bibinfo{volume}{85} (\bibinfo{year}{1998})
  \bibinfo{pages}{549--559}.
\bibitem[{Chan and He(1993)}]{cha}
\bibinfo{author}{Y.M. Chan}, \bibinfo{author}{X.~He}, \bibinfo{title}{On
  median-type estimators of direction for the von {M}ises--{F}isher
  distribution}, \bibinfo{journal}{Biometrika} \bibinfo{volume}{80}
  (\bibinfo{year}{1993}) \bibinfo{pages}{869--875}.
\bibitem[{Ducharme and Milasevic(1987)}]{duc1987}
\bibinfo{author}{G.R. Ducharme}, \bibinfo{author}{P.~Milasevic},
  \bibinfo{title}{Spatial median and directional data},
  \bibinfo{journal}{Biometrika} \bibinfo{volume}{74} (\bibinfo{year}{1987})
  \bibinfo{pages}{212--215}.
\bibitem[{Ducharme and Milasevic(1990)}]{duc1990}
\bibinfo{author}{G.R. Ducharme}, \bibinfo{author}{P.~Milasevic},
  \bibinfo{title}{Estimating the concentration of the {L}angevin distribution},
  \bibinfo{journal}{Canadian Journal of Statistics} \bibinfo{volume}{18}
  (\bibinfo{year}{1990}) \bibinfo{pages}{163--169}.
\bibitem[{Fisher(1982)}]{fis1982}
\bibinfo{author}{N.I. Fisher}, \bibinfo{title}{Robust estimation of the
  concentration parameter of {F}isher's distribution on the sphere.},
  \bibinfo{journal}{Applied Statistics} \bibinfo{volume}{31}
  (\bibinfo{year}{1982}) \bibinfo{pages}{152--154}.
\bibitem[{Fisher(1985)}]{fis1985}
\bibinfo{author}{N.I. Fisher}, \bibinfo{title}{Spherical medians},
  \bibinfo{journal}{Journal of the Royal Statistical Soceity, Series B}
  \bibinfo{volume}{47} (\bibinfo{year}{1985}) \bibinfo{pages}{342--348}.
\bibitem[{Fisher(1993)}]{fis1993}
\bibinfo{author}{N.I. Fisher}, \bibinfo{title}{{Statistical Analysis of
  Circular Data}}, \bibinfo{publisher}{Cambridge University Press},
  \bibinfo{address}{Cambridge}, \bibinfo{year}{1993}.
\bibitem[{Fisher et~al.(1987)Fisher, Lewis and Embleton}]{fis1987}
\bibinfo{author}{N.I. Fisher}, \bibinfo{author}{T.~Lewis},
  \bibinfo{author}{B.J.J. Embleton}, \bibinfo{title}{{Statistical Analysis of
  Spherical Data}}, \bibinfo{publisher}{Cambridge University Press},
  \bibinfo{address}{Cambridge}, \bibinfo{year}{1987}.
\bibitem[{Fujisawa and Eguchi(2006)}]{fuj2006}
\bibinfo{author}{H.~Fujisawa}, \bibinfo{author}{S.~Eguchi},
  \bibinfo{title}{Robust estimation in the normal mixture model},
  \bibinfo{journal}{Journal of Statistical Planning and Inference}
  \bibinfo{volume}{136} (\bibinfo{year}{2006}) \bibinfo{pages}{3989--4011}.
\bibitem[{Fujisawa and Eguchi(2008)}]{fuj2008}
\bibinfo{author}{H.~Fujisawa}, \bibinfo{author}{S.~Eguchi},
  \bibinfo{title}{Robust parameter estimation with a small bias against heavy
  contamination}, \bibinfo{journal}{Journal of Multivariate Analysis}
  \bibinfo{volume}{99} (\bibinfo{year}{2008}) \bibinfo{pages}{2053--2081}.
\bibitem[{Gradshteyn and Ryzhik(2007)}]{gra}
\bibinfo{author}{I.S. Gradshteyn}, \bibinfo{author}{I.M. Ryzhik},
  \bibinfo{title}{{Table of Integrals, Series, and Products, {\rm 7th ed.}}},
  \bibinfo{publisher}{Academic Press}, \bibinfo{address}{San Diego},
  \bibinfo{year}{2007}.
\bibitem[{Hampel et~al.(1986)Hampel, Ronchetti, Rousseeuw and Stahel}]{ham}
\bibinfo{author}{F.~Hampel}, \bibinfo{author}{E.~Ronchetti},
  \bibinfo{author}{P.~Rousseeuw}, \bibinfo{author}{W.~Stahel},
  \bibinfo{title}{{Robust Statistics: the Approach Based on Influence
  Functions.}}, \bibinfo{publisher}{Wiley}, \bibinfo{address}{New York},
  \bibinfo{year}{1986}.
\bibitem[{Hastie et~al.(2009)Hastie, Tibshirani and Friedman}]{has}
\bibinfo{author}{T.~Hastie}, \bibinfo{author}{R.~Tibshirani},
  \bibinfo{author}{J.~Friedman}, \bibinfo{title}{{The Elements of Statistical
  Learning, {\rm 2nd ed.}}}, \bibinfo{publisher}{Springer},
  \bibinfo{address}{New York}, \bibinfo{year}{2009}.
\bibitem[{Huber(1964)}]{hub}
\bibinfo{author}{P.J. Huber}, \bibinfo{title}{Robust estimation of a location
  parameter}, \bibinfo{journal}{Annals of Mathematical Statistics}
  \bibinfo{volume}{35} (\bibinfo{year}{1964}) \bibinfo{pages}{73--101}.
\bibitem[{Jammalamadaka and SenGupta(2001)}]{jam}
\bibinfo{author}{S.R. Jammalamadaka}, \bibinfo{author}{A.~SenGupta},
  \bibinfo{title}{{Topics in Circular Statistics}}, \bibinfo{publisher}{World
  Scientific Publ.}, \bibinfo{address}{Singapore}, \bibinfo{year}{2001}.
\bibitem[{Johnson and Wehrly(1977)}]{joh}
\bibinfo{author}{R.A. Johnson}, \bibinfo{author}{T.E. Wehrly},
  \bibinfo{title}{Measures and models for angular correlation and
  angular-linear correlation}, \bibinfo{journal}{Journal of the Royal
  Statistical Society, Series B} \bibinfo{volume}{39} (\bibinfo{year}{1977})
  \bibinfo{pages}{222--229}.
\bibitem[{Jones et~al.(2001)Jones, Hjort, Harris and Basu}]{jon2001}
\bibinfo{author}{M.C. Jones}, \bibinfo{author}{N.L. Hjort},
  \bibinfo{author}{I.R. Harris}, \bibinfo{author}{A.~Basu}, \bibinfo{title}{A
  comparison of related density--based minimum divergence estimators},
  \bibinfo{journal}{Biometrika} \bibinfo{volume}{88} (\bibinfo{year}{2001})
  \bibinfo{pages}{865--873}.
\bibitem[{Jones and Pewsey(2005)}]{jon2005}
\bibinfo{author}{M.C. Jones}, \bibinfo{author}{A.~Pewsey}, \bibinfo{title}{A
  family of symmetric distributions on the circle}, \bibinfo{journal}{Journal
  of the American Statistical Association} \bibinfo{volume}{100}
  (\bibinfo{year}{2005}) \bibinfo{pages}{1422--1428}.
\bibitem[{Kent(1982)}]{ken}
\bibinfo{author}{J.T. Kent}, \bibinfo{title}{The {F}isher--{B}ingham
  distribution on the sphere}, \bibinfo{journal}{Journal of the Royal
  Statistical Society, Series B} \bibinfo{volume}{44} (\bibinfo{year}{1982})
  \bibinfo{pages}{71--80}.
\bibitem[{Ko(1992)}]{ko}
\bibinfo{author}{D.~Ko}, \bibinfo{title}{Robust estimation of the concentration
  parameter of the von {M}ises--{F}isher distribution},
  \bibinfo{journal}{Annals of Statistics} \bibinfo{volume}{20}
  (\bibinfo{year}{1992}) \bibinfo{pages}{917--928}.
\bibitem[{Ko and Guttorp(1988)}]{kog}
\bibinfo{author}{D.~Ko}, \bibinfo{author}{P.~Guttorp},
  \bibinfo{title}{Robustness of estimators for directional data},
  \bibinfo{journal}{Annals of Statistics} \bibinfo{volume}{16}
  (\bibinfo{year}{1988}) \bibinfo{pages}{609--618}.
\bibitem[{Lenth(1981)}]{len}
\bibinfo{author}{R.V. Lenth}, \bibinfo{title}{Robust measure of location of
  location for directional data}, \bibinfo{journal}{Technometrics}
  \bibinfo{volume}{23} (\bibinfo{year}{1981}) \bibinfo{pages}{77--81}.
\bibitem[{Mardia(1972)}]{mar1972}
\bibinfo{author}{K.V. Mardia}, \bibinfo{title}{{Statistics of Directional
  Data}}, \bibinfo{publisher}{Academic Press}, \bibinfo{address}{New York},
  \bibinfo{year}{1972}.
\bibitem[{Mardia and Jupp(1999)}]{mar1999}
\bibinfo{author}{K.V. Mardia}, \bibinfo{author}{P.E. Jupp},
  \bibinfo{title}{{Directional Statistics}}, \bibinfo{publisher}{Wiley},
  \bibinfo{address}{Chichester}, \bibinfo{year}{1999}.
\bibitem[{Minami and Eguchi(2002)}]{min}
\bibinfo{author}{M.~Minami}, \bibinfo{author}{S.~Eguchi},
  \bibinfo{title}{Robust blind source separation by beta divergence},
  \bibinfo{journal}{Neural Computation} \bibinfo{volume}{14}
  (\bibinfo{year}{2002}) \bibinfo{pages}{1859--1886}.
\bibitem[{Stephens(1979)}]{ste}
\bibinfo{author}{M.A. Stephens}, \bibinfo{title}{Vector correlation},
  \bibinfo{journal}{Biometrika} \bibinfo{volume}{66} (\bibinfo{year}{1979})
  \bibinfo{pages}{41--48}.
\bibitem[{Watson(1983)}]{wat1983}
\bibinfo{author}{G.S. Watson}, \bibinfo{title}{Statistics on Spheres},
  \bibinfo{publisher}{Wiley}, \bibinfo{address}{New York},
  \bibinfo{year}{1983}.
\bibitem[{Wehrly and Shine(1981)}]{weh}
\bibinfo{author}{T.E. Wehrly}, \bibinfo{author}{E.P. Shine},
  \bibinfo{title}{Influence curves of estimators for directional data},
  \bibinfo{journal}{Biometrika} \bibinfo{volume}{68} (\bibinfo{year}{1981})
  \bibinfo{pages}{334--335}.
\bibitem[{Windham(1994)}]{win}
\bibinfo{author}{P.M. Windham}, \bibinfo{title}{Robustifying model fitting},
  \bibinfo{journal}{Journal of the Royal Statistical Society, Series B}
  \bibinfo{volume}{57} (\bibinfo{year}{1994}) \bibinfo{pages}{599--609}.

\end{thebibliography}






\newpage

\baselineskip 8.5mm

\noindent Table 1. \textit{Estimates of the relative mean squared errors of the minimum divergence estimators with respect to maximum likelihood estimator for some selected sample sizes and tuning parameters for the von Mises--Fisher distribution $\mbox{vM}_p(\xi)$ with $p=2$ and $\xi=(2.37,0)'$ and $p=3$ and $\xi=(3.99,0,0)'$.}
\vspace{0.5cm}

\noindent \hspace{-2.7cm} {\footnotesize
\begin{tabular}{cccccccccccc}
 & & & & $p=2$ & & & & & $p=3$ & & \\[2mm]
 & & $n=10$ & $n=20$ & $n=30$ & $n=50$ & $n=100$ & $n=10$ & $n=20$ & $n=30$ & $n=50$ & $n=100$ \\[2mm]
Type 1 & $\beta=0.02$ & 0.994 & 0.996 & 0.996 & 1.000 & 0.999 & 0.991 & 0.996 & 0.996 & 0.999 & 1.000 \\[2mm]
 & $\beta=0.05$ & 0.993 & 0.995 & 0.995 & 1.003 & 1.001 & 0.985 & 0.997 & 0.996 & 1.003 & 1.003 \\[2mm]
 & $\beta=0.1$  & 1.016 & 1.012 & 1.008 & 1.021 & 1.012 & 1.006 & 1.019 & 1.008 & 1.021 & 1.019 \\[2mm]
 & $\beta=0.25$ & 5.464 & 1.231 & 1.154 & 1.151 & 1.104 & 1.465 & 1.213 & 1.132 & 1.146 & 1.126 \\[2mm]
 & $\beta=0.5$  & 14.836 & 2.357 & 1.711 & 1.559 & 1.393 & 6.512 & 1.969 & 1.568 & 1.515 & 1.418 \\[2mm]
 & $\beta=0.75$ & 21.547 & 4.314 & 2.735 & 2.232 & 1.723 & 19.290 & 3.028 & 2.159 & 1.994 & 1.752 \\[2mm]
Type 0 & $\gamma=0.02$ & 0.994 & 0.996 & 0.996 & 1.000 & 0.999 & 0.991 & 0.996 & 0.996 & 0.999 & 1.000 \\[2mm]
 & $\gamma=0.05$ & 0.993 & 0.995 & 0.995 & 1.003 & 1.001 & 0.985 & 0.997 & 0.996 & 1.003 & 1.003 \\[2mm]
 & $\gamma=0.1$  & 1.016 & 1.013 & 1.008 & 1.021 & 1.012 & 1.007 & 1.019 & 1.008 & 1.021 & 1.019 \\[2mm]
 & $\gamma=0.25$ & 6.179 & 1.247 & 1.164 & 1.158 & 1.109 & 1.618 & 1.234 & 1.143 & 1.156 & 1.133 \\[2mm]
 & $\gamma=0.5$  & 81.673 & 3.501 & 2.036 & 1.702 & 1.479 & 458.822 &3.312 & 1.860 & 1.697 & 1.530 \\[2mm]
 & $\gamma=0.75$ & 737.175 & 58.871 & 137.763 & 4.240 & 2.135 & 2892.802 & 820.520 & 178.742 & 4.206 & 2.285 \\[2mm]
\end{tabular}
}

\newpage

\noindent Table 2. \ \textit{Estimates of the relative mean squared errors of the minimum divergence estimators with respect to maximum likelihood estimator for some selected contamination sizes and tuning parameters for a mixture of the von Mises--Fisher and the uniform distributions $ (1-\varepsilon)\, \mbox{vM}_p(\xi) + \varepsilon \, U_p $ with (a) $p=2,\ \xi=(0.52,0)'$ and $p=3,\ \xi=(0.78,0,0)'$, (b) $p=2,\ \xi=(1.16,0)'$ and $p=3,\ \xi=(1.80,0,0)'$, (c) $p=2,\ \xi=(2.37,0)'$ and $p=3,\ \xi=(3.99,0,0)'$, and (d) $p=2,\ \xi=(10.27,0)'$ and $p=3,\ \xi=(20.0,0,0)'$, for each $\varepsilon=0.02,\ 0.05,\ 0.1$ and $0.2$.}\vspace{0.5cm}

\begin{center}
(a) $p=2$ and $\xi=(0.52,0)'$ and $p=3$ and $\xi=(0.78,0,0)'$ \vspace{0.2cm}

\end{center}
\hspace{-1.2cm} {\footnotesize 
\begin{tabular}{cccccccccc}
 & & & \multicolumn{2}{c}{$p=2$} & & & \multicolumn{2}{c}{$p=3$}& \\[2mm]
 & & $\varepsilon=0.02$ & $\varepsilon=0.05$ & $\varepsilon=0.1$ & $\varepsilon=0.2$ & $\varepsilon=0.02$ & $\varepsilon=0.05$ & $\varepsilon=0.1$ & $\varepsilon=0.2$ \\[2mm]
Type 1 & $\beta=0.02$ & 1.000 & 1.000 & 1.000 & 1.000 & 1.000 & 1.000 & 1.000 & 1.000 \\[2mm]
 & $\beta=0.05$ & 1.001 & 1.001 & 1.000 & 1.000 & 1.000 & 1.001 & 1.000 & 0.999 \\[2mm]
 & $\beta=0.1$  & 1.002 & 1.002 & 1.001 & 1.001 & 1.001 & 1.002 & 1.001 & 0.999 \\[2mm]
 & $\beta=0.25$ & 1.009 & 1.009 & 1.006 & 1.004 & 1.011 & 1.012 & 1.008 & 1.001 \\[2mm]
 & $\beta=0.5$  & 1.030 & 1.030 & 1.021 & 1.013 & 1.047 & 1.047 & 1.035 & 1.014  \\[2mm]
 & $\beta=0.75$ & 1.064 & 1.062 & 1.046 & 1.028 & 1.107 & 1.106 & 1.081 & 1.040  \\[2mm]
Type 0 & $\gamma=0.02$ & 1.000 & 1.000 & 1.000 & 1.000 & 1.000 & 1.000 & 1.000 & 1.000  \\[2mm]
 & $\gamma=0.05$ & 1.001 & 1.001 & 1.000 & 1.000 & 1.000 & 1.001 & 1.000 & 0.999 \\[2mm]
 & $\gamma=0.1$  & 1.002 & 1.002 & 1.001 & 1.001 & 1.001 & 1.002 & 1.001 & 0.999 \\[2mm]
 & $\gamma=0.25$ & 1.009 & 1.009 & 1.006 & 1.004 & 1.011 & 1.012 & 1.008 & 1.001 \\[2mm]
 & $\gamma=0.5$  & 1.031 & 1.031 & 1.022 & 1.013 & 1.049 & 1.050 & 1.037 & 1.015 \\[2mm]
 & $\gamma=0.75$ & 1.071 & 1.070 & 1.051 & 1.030 & 1.122 & 1.122 & 1.093 & 1.046
\end{tabular}
}

\newpage 
\begin{center}
(b) $p=2$ and $\xi=(1.16,0)'$ and $p=3$ and $\xi=(1.80,0,0)'$
\end{center}\vspace{0.2cm}
\hspace{-1cm} {\footnotesize
\begin{tabular}{cccccccccc}
 & & & \multicolumn{2}{c}{$p=2$} & & & \multicolumn{2}{c}{$p=3$}& \\[2mm]
 & & $\varepsilon=0.02$ & $\varepsilon=0.05$ & $\varepsilon=0.1$ & $\varepsilon=0.2$ & $\varepsilon=0.02$ & $\varepsilon=0.05$ & $\varepsilon=0.1$ & $\varepsilon=0.2$ \\[2mm]
Type 1 & $\beta=0.02$ & 1.000 & 1.000 & 0.999 & 0.998 & 1.000 & 0.999 & 0.996 & 0.994 \\[2mm]
 & $\beta=0.05$ & 1.001 & 0.999 & 0.997 & 0.995 & 1.001 & 0.998 & 0.991 & 0.985 \\[2mm]
 & $\beta=0.1$  & 1.004 & 1.000 & 0.995 & 0.990 & 1.006 & 0.998 & 0.983 & 0.970 \\[2mm]
 & $\beta=0.25$ & 1.024 & 1.012 & 0.994 & 0.977 & 1.040 & 1.016 & 0.973 & 0.930 \\[2mm]
 & $\beta=0.5$  & 1.098 & 1.065 & 1.015 & 0.960 & 1.161 & 1.100 & 1.002 & 0.879 \\[2mm]
 & $\beta=0.75$ & 1.213 & 1.155 & 1.063 & 0.952 & 1.337 & 1.232 & 1.079 & 0.855 \\[2mm]
Type 0 & $\gamma=0.02$ & 1.000 & 1.000 & 0.999 & 0.998 & 1.000 & 0.999 & 0.996 & 0.994 \\[2mm]
 & $\gamma=0.05$ & 1.001 & 0.999 & 0.997 & 0.995 & 1.001 & 0.998 & 0.991 & 0.985 \\[2mm]
 & $\gamma=0.1$  & 1.004 & 1.000 & 0.995 & 0.990 & 1.006 & 0.998 & 0.983 & 0.970 \\[2mm]
 & $\gamma=0.25$ & 1.025 & 1.013 & 0.994 & 0.976 & 1.042 & 1.017 & 0.974 & 0.929 \\[2mm]
 & $\gamma=0.5$  & 1.111 & 1.074 & 1.020 & 0.959 & 1.190 & 1.120 & 1.012 & 0.874 \\[2mm]
 & $\gamma=0.75$ & 1.296 & 1.216 & 1.103 & 0.960 & 1.515 & 1.355 & 1.173 & 0.854 
\end{tabular}\vspace{0.5cm}
}

\newpage

\begin{center}
(c) $p=2$ and $\xi=(2.37,0)'$ and $p=3$ and $\xi=(3.99,0,0)'$
\end{center}\vspace{0.2cm}
\hspace{-1cm} {\footnotesize
\begin{tabular}{cccccccccc}
 & & & \multicolumn{2}{c}{$p=2$} & & & \multicolumn{2}{c}{$p=3$}& \\[2mm]
 & & $\varepsilon=0.02$ & $\varepsilon=0.05$ & $\varepsilon=0.1$ & $\varepsilon=0.2$ & $\varepsilon=0.02$ & $\varepsilon=0.05$ & $\varepsilon=0.1$ & $\varepsilon=0.2$ \\[2mm]
Type 1 & $\beta=0.02$ & 0.999 & 0.991 & 0.986 & 0.989 & 0.990 & 0.967 & 0.962 & 0.974 \\[2mm]
 & $\beta=0.05$ & 1.001 & 0.980 & 0.964 & 0.972 & 0.982 & 0.922 & 0.907 & 0.935 \\[2mm]
 & $\beta=0.1$  & 1.011 & 0.965 & 0.930 & 0.943 & 0.980 & 0.862 & 0.822 & 0.869 \\[2mm]
 & $\beta=0.25$ & 1.097 & 0.961 & 0.840 & 0.857 & 1.058 & 0.777 & 0.628 & 0.682 \\[2mm]
 & $\beta=0.5$  & 1.386 & 1.084 & 0.761 & 0.726 & 1.325 & 0.838 & 0.508 & 0.469 \\[2mm]
 & $\beta=0.75$ & 1.736 & 1.289 & 0.760 & 0.640 & 1.616 & 0.968 & 0.511 & 0.391 \\[2mm]
Type 0 & $\gamma=0.02$ & 0.999 & 0.991 & 0.986 & 0.989 & 0.990 & 0.967 & 0.962 & 0.974 \\[2mm]
 & $\gamma=0.05$ & 1.001 & 0.980 & 0.964 & 0.972 & 0.982 & 0.922 & 0.907 & 0.935 \\[2mm]
 & $\gamma=0.1$  & 1.011 & 0.965 & 0.929 & 0.943 & 0.980 & 0.861 & 0.821 & 0.869 \\[2mm]
 & $\gamma=0.25$ & 1.102 & 0.962 & 0.837 & 0.854 & 1.065 & 0.776 & 0.619 & 0.673 \\[2mm]
 & $\gamma=0.5$  & 1.480 & 1.141 & 0.760 & 0.705 & 1.446 & 0.900 & 0.502 & 0.421 \\[2mm]
 & $\gamma=0.75$ & 2.245 & 1.660 & 0.851 & 0.619 & 2.158 & 1.308 & 0.616 & 0.342
\end{tabular}
}

\newpage

\begin{center}
(d) $p=2$ and $\xi=(10.27,0)'$ and $p=3$ and $\xi=(20.00,0,0)'$
\end{center}\vspace{0.2cm}
\hspace{-1cm} {\footnotesize
\begin{tabular}{cccccccccc}
 & & & \multicolumn{2}{c}{$p=2$} & & & \multicolumn{2}{c}{$p=3$}& \\[2mm]
 & & $\varepsilon=0.02$ & $\varepsilon=0.05$ & $\varepsilon=0.1$ & $\varepsilon=0.2$ & $\varepsilon=0.02$ & $\varepsilon=0.05$ & $\varepsilon=0.1$ & $\varepsilon=0.2$ \\[2mm]
Type 1 & $\beta=0.02$ & 0.841 & 0.891 & 0.945 & 0.983 & 0.691 & 0.792 & 0.894 & 0.967 \\[2mm]
 & $\beta=0.05$ & 0.635 & 0.723 & 0.851 & 0.953 & 0.372 & 0.486 & 0.700 & 0.903 \\[2mm]
 & $\beta=0.1$  & 0.412 & 0.461 & 0.665 & 0.891 & 0.194 & 0.172 & 0.332 & 0.733 \\[2mm]
 & $\beta=0.25$ & 0.297 & 0.157 & 0.198 & 0.535 & 0.173 & 0.071 & 0.054 & 0.114 \\[2mm]
 & $\beta=0.5$  & 0.365 & 0.153 & 0.109 & 0.164 & 0.220 & 0.081 & 0.049 & 0.065 \\[2mm]
 & $\beta=0.75$ & 0.452 & 0.180 & 0.115 & 0.144 & 0.271 & 0.095 & 0.057 & 0.075 \\[2mm]
Type 0 & $\gamma=0.02$ & 0.841 & 0.891 & 0.945 & 0.983 & 0.691 & 0.792 & 0.894 & 0.967 \\[2mm]
 & $\gamma=0.05$ & 0.634 & 0.723 & 0.851 & 0.953 & 0.371 & 0.485 & 0.699 & 0.903 \\[2mm]
 & $\gamma=0.1$  & 0.411 & 0.459 & 0.663 & 0.890 & 0.193 & 0.168 & 0.326 & 0.730 \\[2mm]
 & $\gamma=0.25$ & 0.300 & 0.154 & 0.186 & 0.508 & 0.177 & 0.070 & 0.048 & 0.081 \\[2mm]
 & $\gamma=0.5$  & 0.400 & 0.169 & 0.109 & 0.114 & 0.254 & 0.094 & 0.053 & 0.042 \\[2mm]
 & $\gamma=0.75$ & 0.593 & 0.245 & 0.145 & 0.111 & 0.407 & 0.146 & 0.085 & 0.065
\end{tabular}\vspace{0.5cm}
}

\newpage
\noindent Table 3. \ \textit{Estimates of the relative mean squared errors of the minimum divergence estimators with respect to maximum likelihood estimator for some selected sample sizes and tuning parameters for a mixture of the von Mises--Fisher distributions $ (1-\varepsilon)\, \mbox{vM}_p(\xi) + \varepsilon \, \mbox{vM}_p (\zeta)$ with (i) $p=2,\ \xi=(2.37,0)'$ and $\zeta=(-100,0)'$ and (ii) $p=3,\ \xi=(3.99,0,0)'$ and $\zeta=(-199,0,0)'$ for each $\varepsilon=0.05,\ 0.1,\ 0.2,$ and $0.3$.} \vspace{0.5cm}

\hspace{-1.8cm} {\footnotesize
\begin{tabular}{cccccccccc}
 & & & \multicolumn{2}{c}{$p=2$} & & & \multicolumn{2}{c}{$p=3$}& \\[2mm]
 & & $\varepsilon=0.02$ & $\varepsilon=0.05$ & $\varepsilon=0.1$ & $\varepsilon=0.2$ & $\varepsilon=0.02$ & $\varepsilon=0.05$ & $\varepsilon=0.1$ & $\varepsilon=0.2$ \\[2mm]
Type 1 & $\beta=0.02$ & 0.982 & 0.969 & 0.975 & 0.988 & 0.939 & 0.923 & 0.945 & 0.975 \\[2mm]
 & $\beta=0.05$ & 0.958 & 0.924 & 0.937 & 0.969 & 0.861 & 0.813 & 0.860 & 0.935  \\[2mm]
 & $\beta=0.1$  & 0.925 & 0.852 & 0.874 & 0.936 & 0.767 & 0.652 & 0.720 & 0.863  \\[2mm]
 & $\beta=0.25$ & 0.894 & 0.674 & 0.687 & 0.829 & 0.697 & 0.374 & 0.362 & 0.601 \\[2mm]
 & $\beta=0.5$  & 1.041 & 0.553 & 0.448 & 0.629 & 0.835 & 0.328 & 0.175 & 0.214 \\[2mm]
 & $\beta=0.75$ & 1.273 & 0.565 & 0.352 & 0.451 & 0.997 & 0.367 & 0.168 & 0.136 \\[2mm]
Type 0 & $\gamma=0.02$ & 0.982 & 0.969 & 0.975 & 0.988 & 0.939 & 0.923 & 0.945 & 0.975  \\[2mm]
 & $\gamma=0.05$ & 0.957 & 0.924 & 0.937 & 0.969 & 0.861 & 0.813 & 0.860 & 0.935 \\[2mm]
 & $\gamma=0.1$  & 0.925 & 0.851 & 0.874 & 0.936 & 0.766 & 0.651 & 0.719 & 0.863  \\[2mm]
 & $\gamma=0.25$ & 0.895 & 0.669 & 0.681 & 0.826 & 0.700 & 0.365 & 0.345 & 0.588 \\[2mm]
 & $\gamma=0.5$  & 1.123 & 0.563 & 0.422 & 0.600 & 0.931 & 0.366 & 0.173 & 0.145 \\[2mm]
 & $\gamma=0.75$ & 1.760 & 0.738 & 0.412 & 0.389 & 1.379 & 0.533 & 0.243 & 0.136 \\[2mm]
\end{tabular} 
}

\newpage

\noindent Table 4. \ \textit{Estimates of the parameters and tuning parameters for the maximum likelihood estimators and two minimum divergence estimators. The maximum likelihood estimators are obtained for some subsets of the data excluding no samples, one at 2.57 and ones at 2.57 and 5.20. The parameters $\hat{\mu}$ and $\hat{\kappa}$ are defined by $ \hat{\mu}=\mbox{Arg} (\cos \hat{\xi}_1 + i \sin \hat{\xi}_2)$ and $\hat{\kappa}= ( \hat{\xi}_1^2 + \hat{\xi}_2^2 )^{1/2} $, respectively, where $\hat{\xi}=(\hat{\xi}_1,\hat{\xi}_2)'$.}
\vspace{0.5cm}

\begin{center}
\hspace{-0.3cm} {\small
\begin{tabular}{cccccc}
 & MLE & MLE (with one & MLE (with two & Type 1 & Type 0 \\[-0.5mm]
 &     & sample excluded) & samples excluded) &        &        \\[1mm]
\hline \tabtopsp{1mm}%
tuning parameter & $-$ & $-$ & $-$ & 0.59 & 0.48 \\[2mm]
$\hat{\mu}$ & 0.0541 & 0.0232 & 0.0712 & 0.0377 & 0.0380 \\[2mm]
$\hat{\kappa}$ & 3.30 & 5.74 & 7.66 & 5.86 & 5.98 
\end{tabular}
}
\end{center}

\newpage

\renewcommand{\thefigure}{\arabic{figure}}

\begin{figure}
\begin{center}
\includegraphics[width=6.8cm,height=4.25cm]{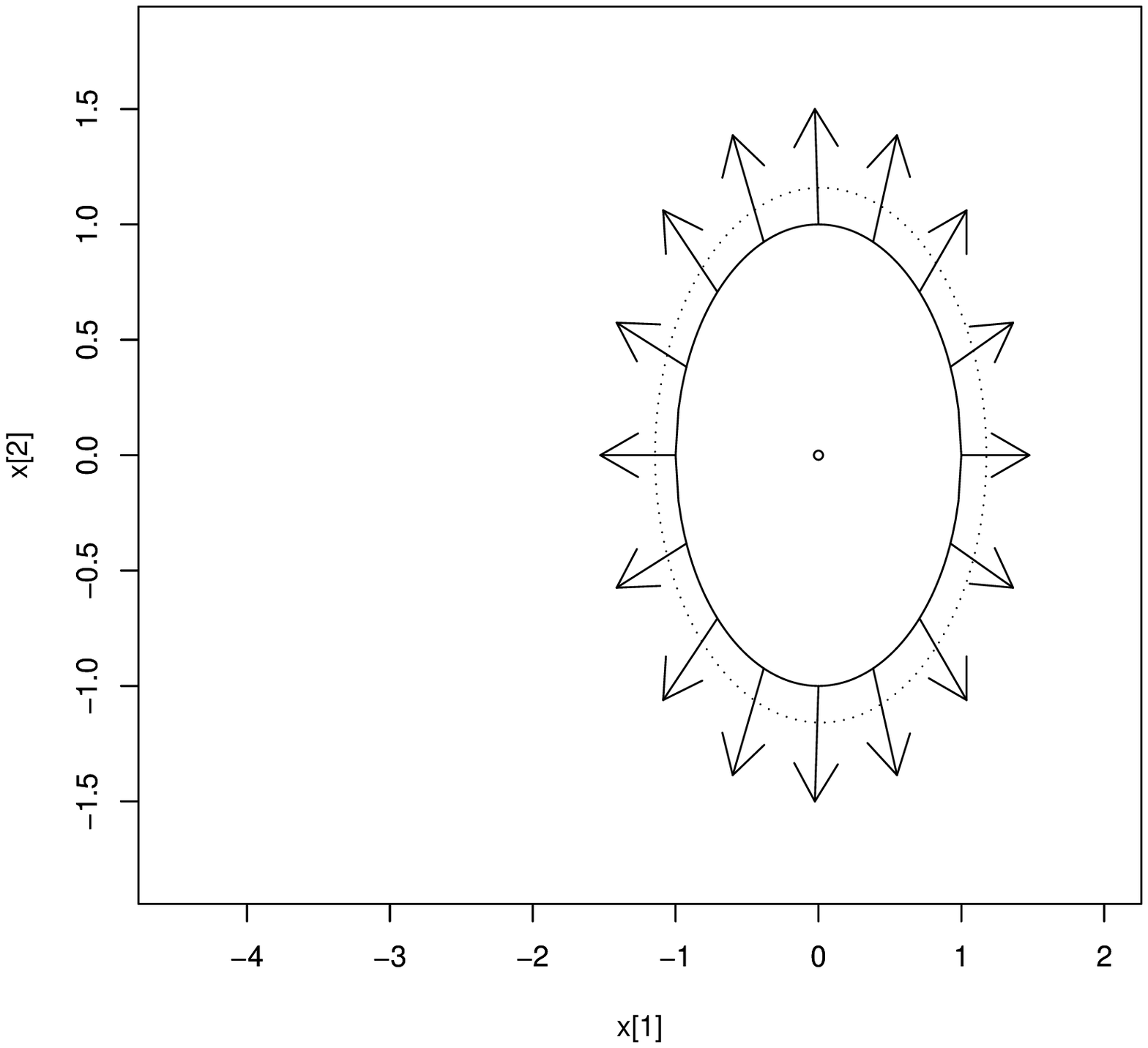}
\includegraphics[width=6.8cm,height=4.25cm]{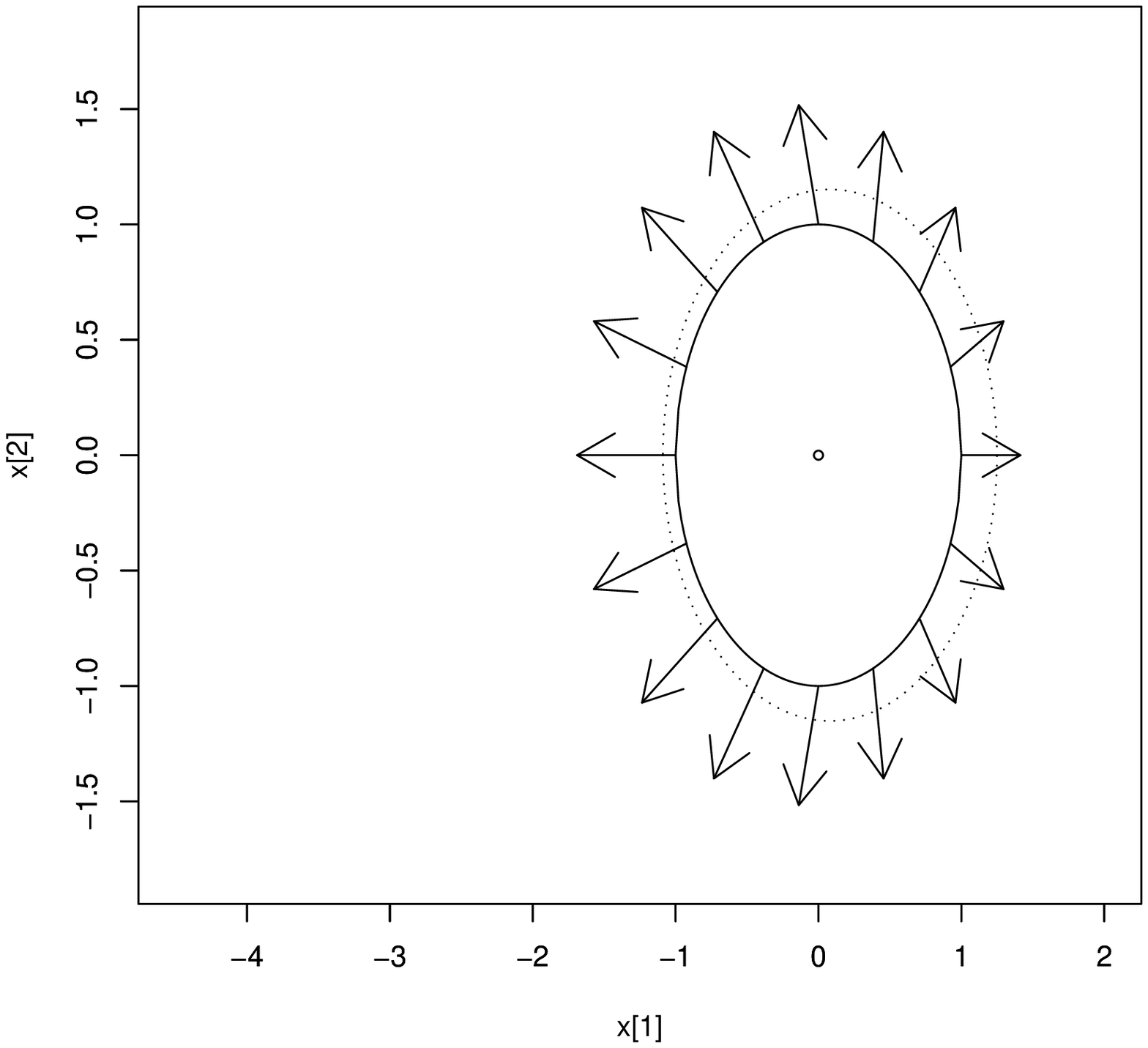}\\
\hspace{0.3cm} (a) \hspace{6.5cm} (b)\\
\includegraphics[width=6.8cm,height=4.25cm]{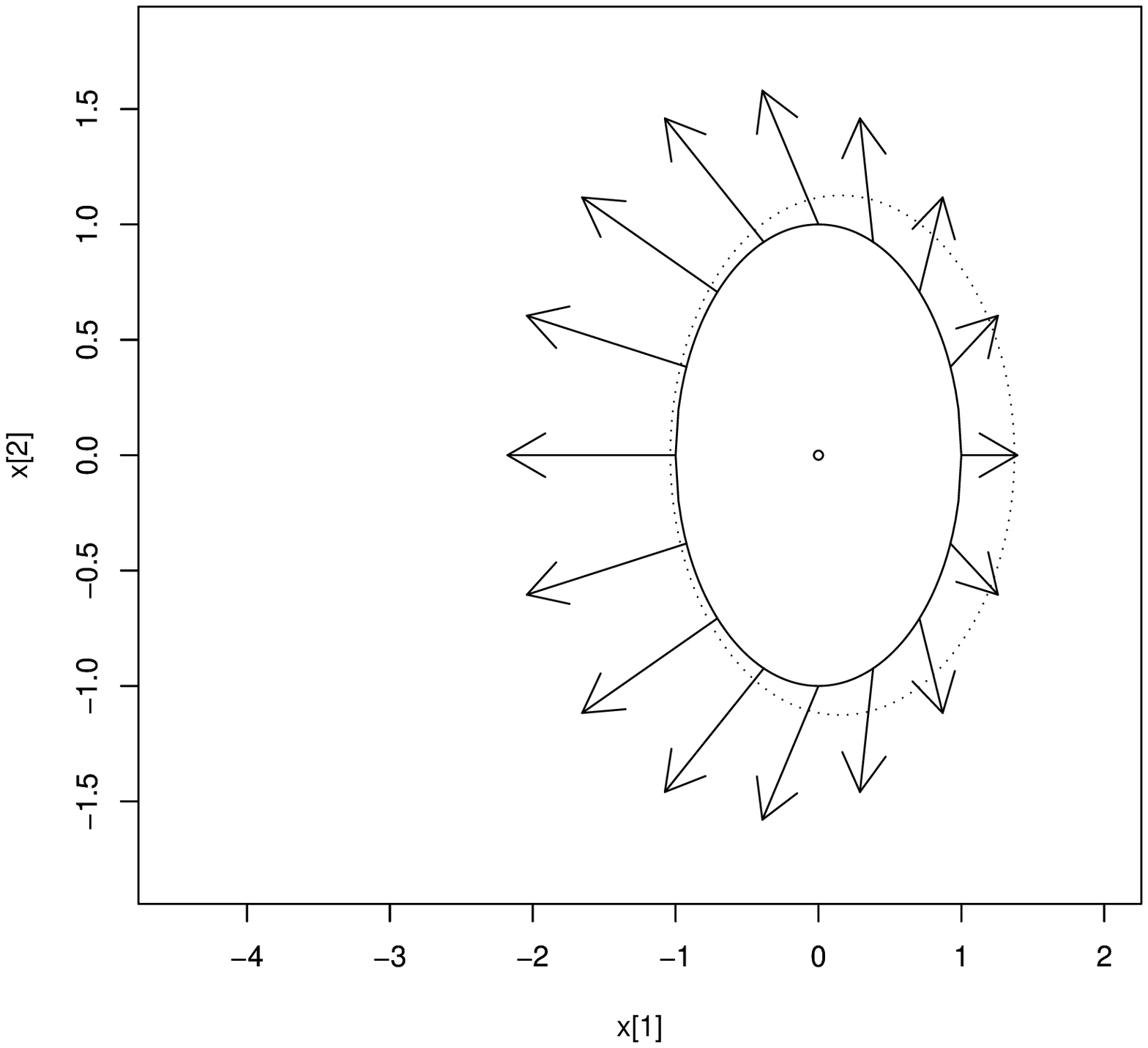}
\includegraphics[width=6.8cm,height=4.25cm]{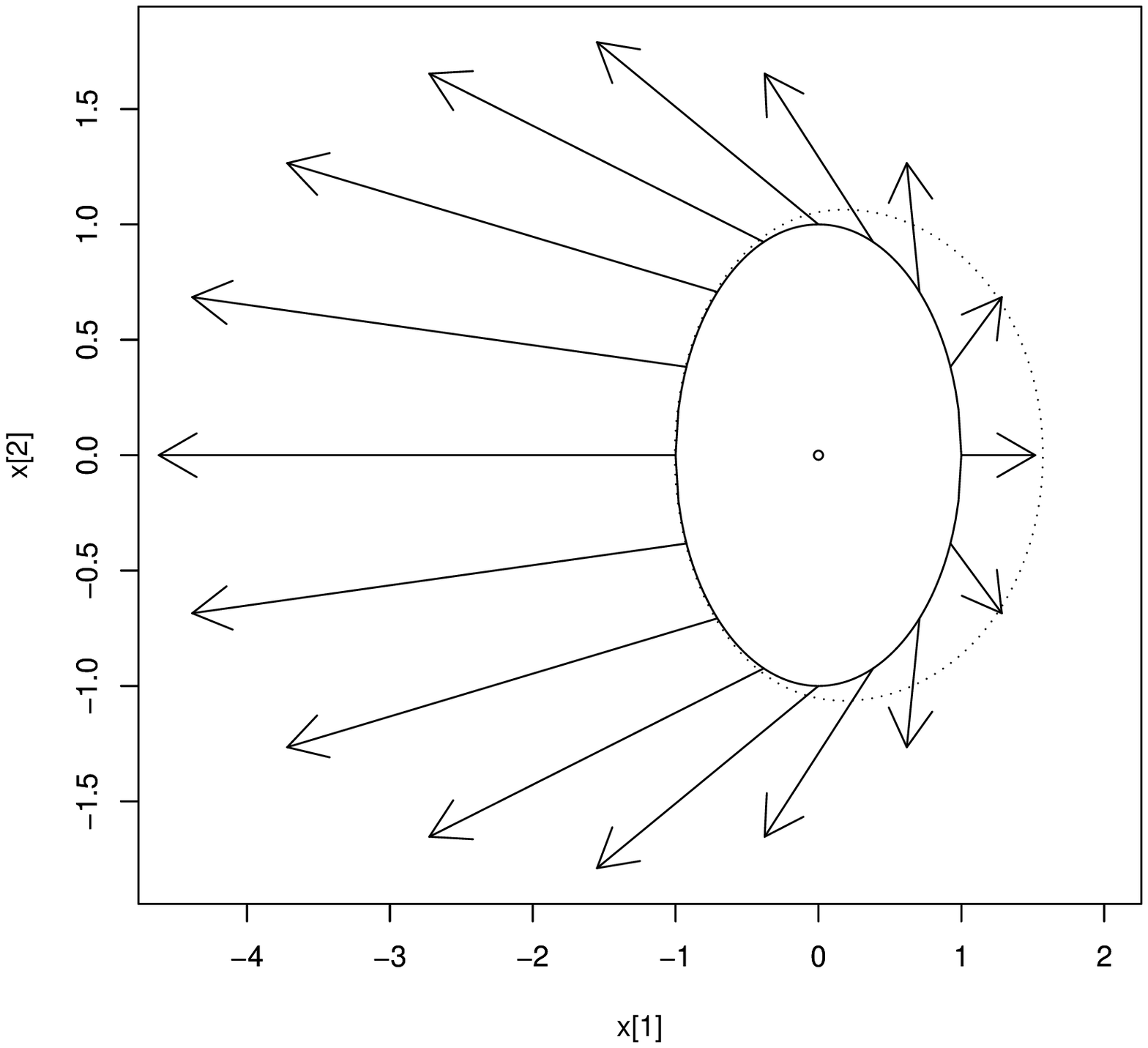}\\
\hspace{0.3cm} (c) \hspace{6.5cm} (d)\\
\end{center}
\caption[]{
\baselineskip 8.5mm
Influence functions (\ref{if_mle}) of maximum likelihood estimators for (a) $\xi=(0.10,0)',$ (b) $\xi=(0.52,0)',$ (c) $\xi=(1.16,0)'$ and (d) $\xi=(2.37,0)'$ for the $\mbox{vM}_2(\xi)$ model.
For convenience, the norms of the influence functions are divided by four.
In each frame, the white dot denotes the origin, while the black one denotes $A_2(\|\xi\|)\,\xi/\|\xi\|$.}
\end{figure}

\newpage

\begin{figure}
\begin{center}
\includegraphics[trim=0cm 0cm 0cm 1cm,clip,width=7cm,height=7cm]{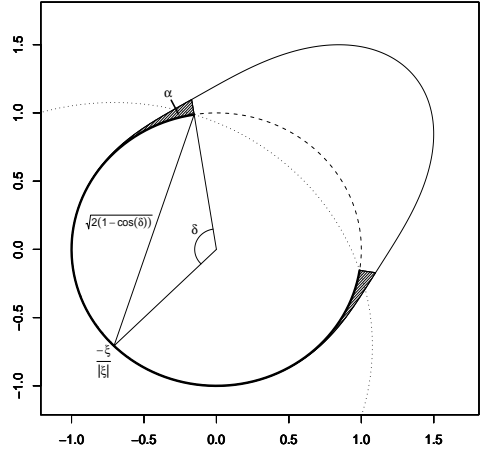} \\
(a)\\
\includegraphics[width=8cm,height=7cm]{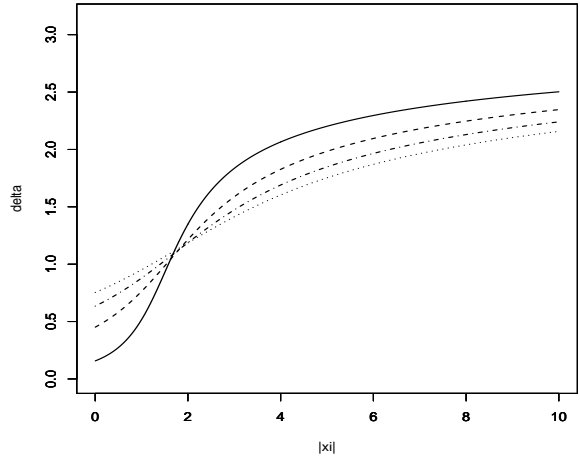} \\
(b)
\caption[]{
\baselineskip 8.5mm
(a) Plot of the von Mises--Fisher density $\mbox{vM}_2 \{ (5,0)' \}$ (solid), the unit circle (dashed), the disc $N$ (dotted) and the area $Ar_2$ (bold and solid) with $\alpha=0.05$ and (b) plot of $\delta$ satisfying Equation (\ref{tail}) with $\alpha=0.05$ and $p=2$ (solid), $3$ (dashed), $4$ (dot-dashed) and $5$ (dotted) as a function of $\|\xi\|$.
}
\label{fig1}
\end{center}
\end{figure}

\newpage

\begin{figure}
\begin{center}
\includegraphics[width=6.8cm,height=4.25cm]{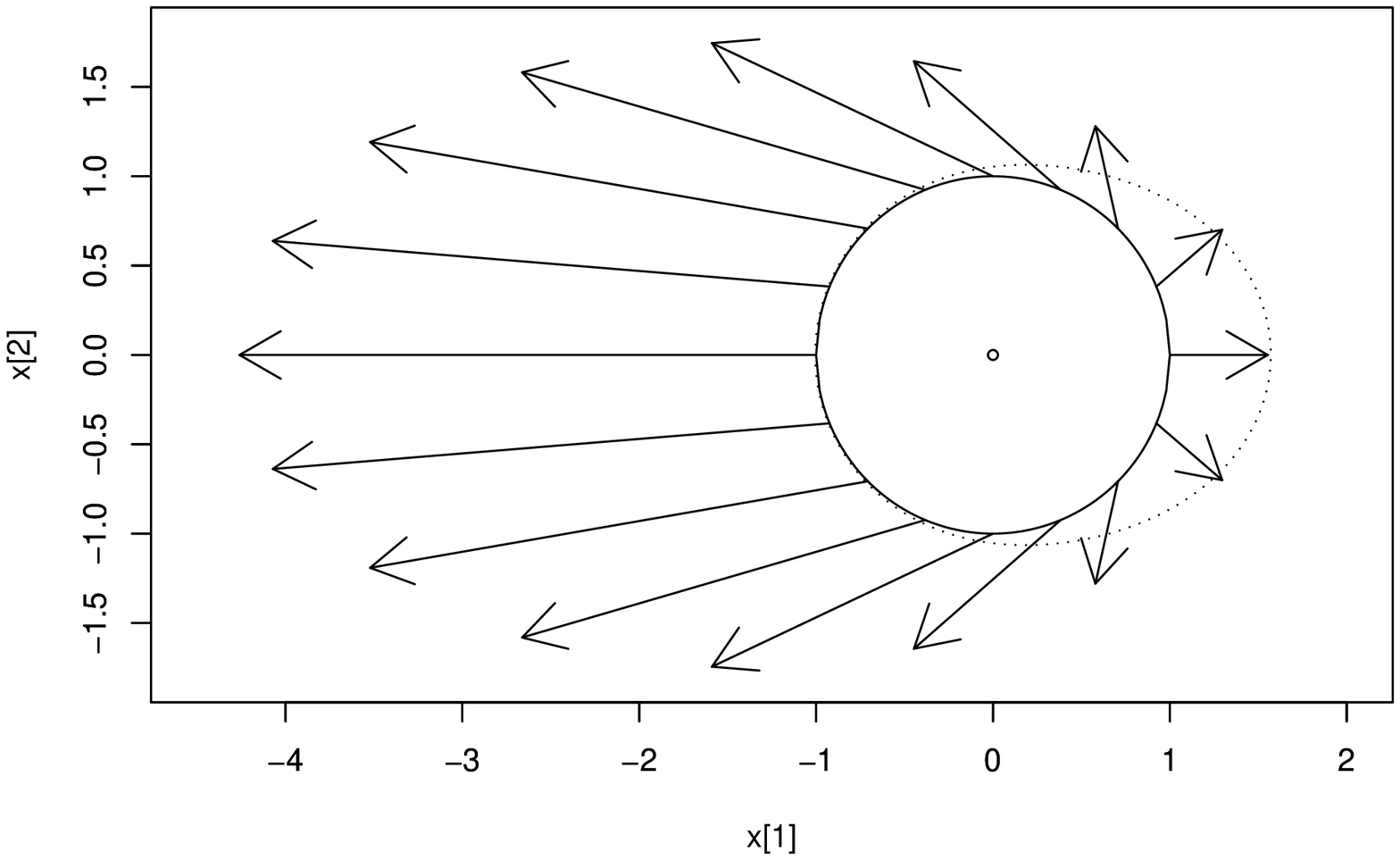}
\includegraphics[width=6.8cm,height=4.25cm]{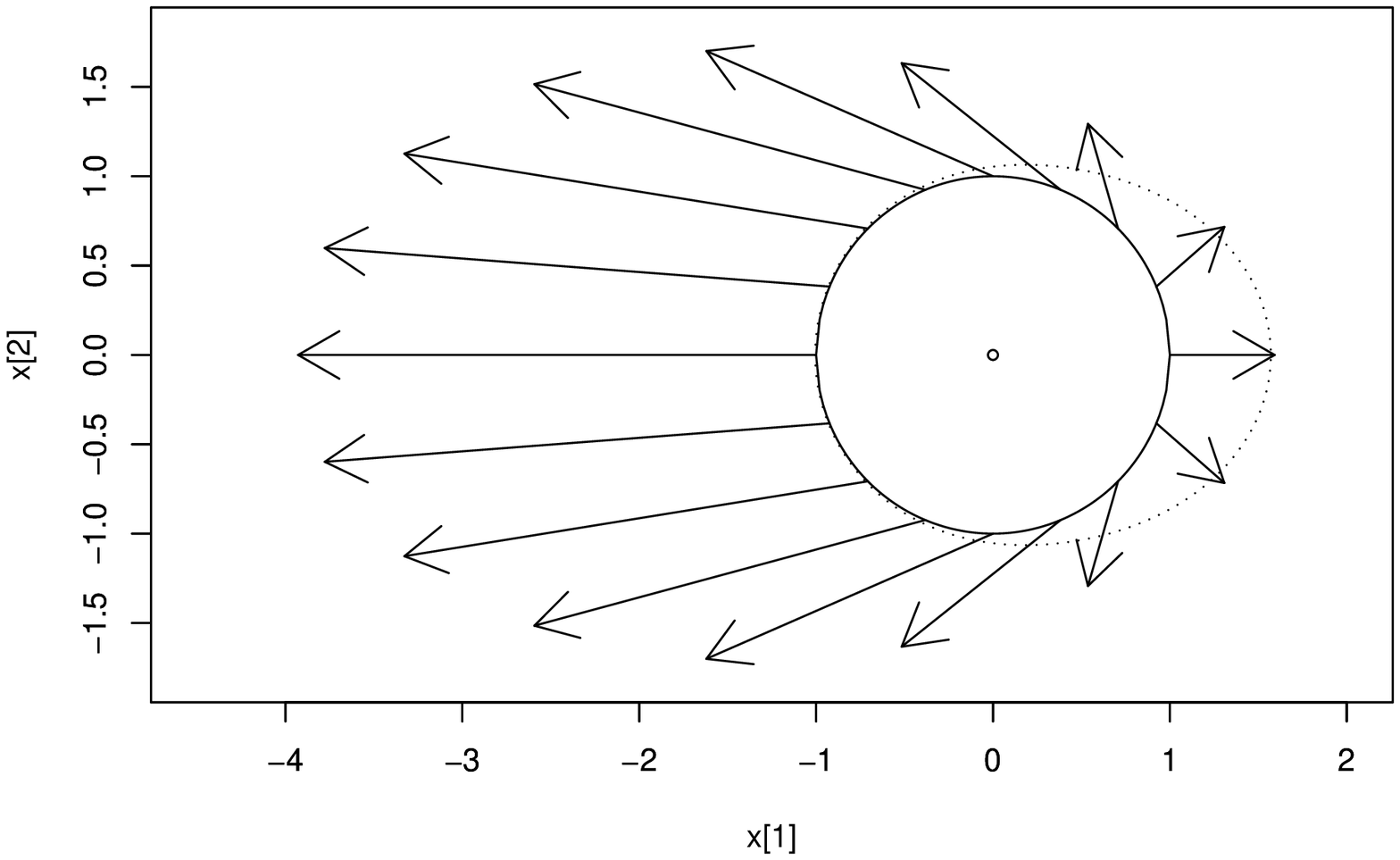}\\
\hspace{0.3cm} (a) \hspace{6.5cm} (b)\\
\includegraphics[width=6.8cm,height=4.25cm]{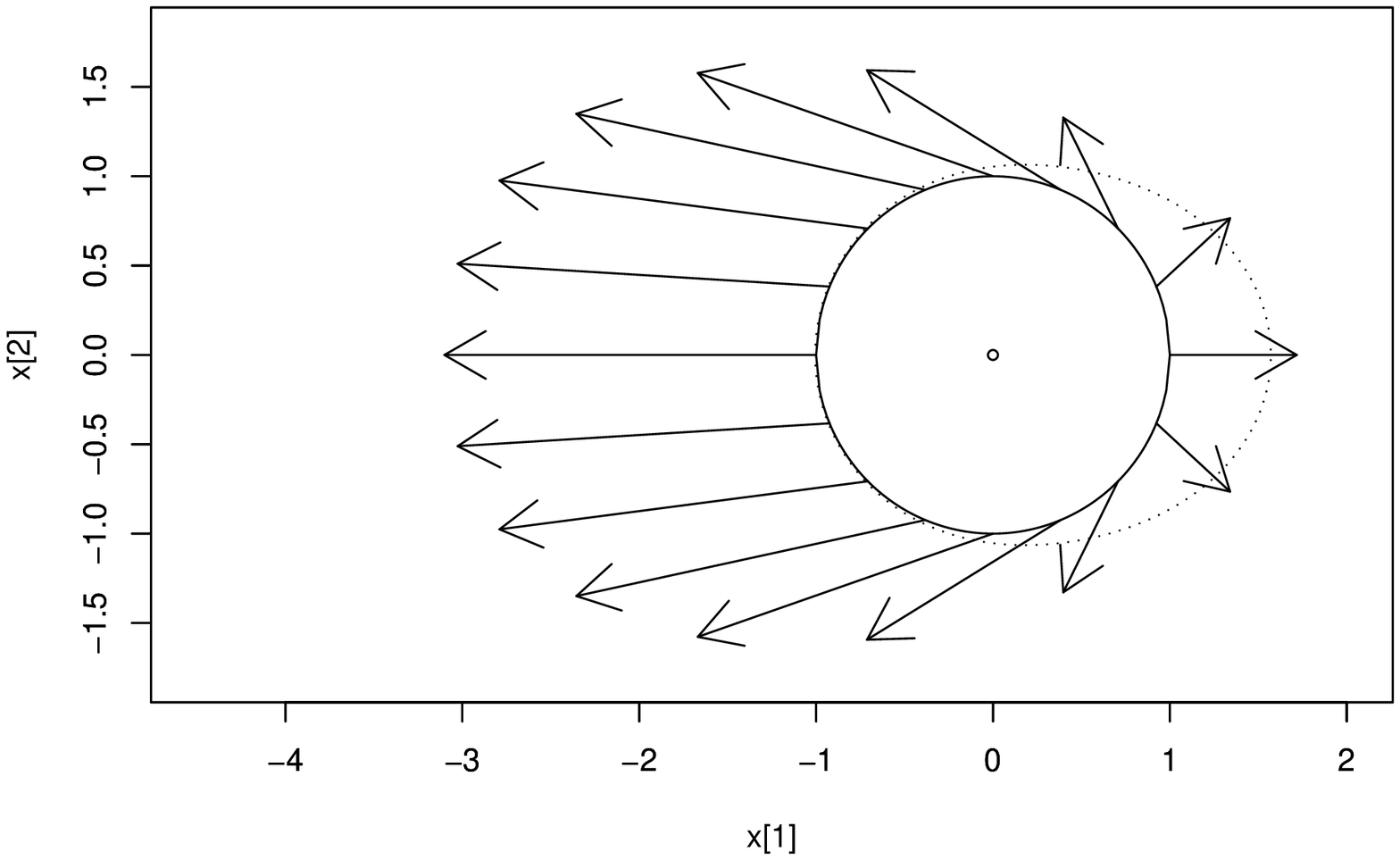}
\includegraphics[width=6.8cm,height=4.25cm]{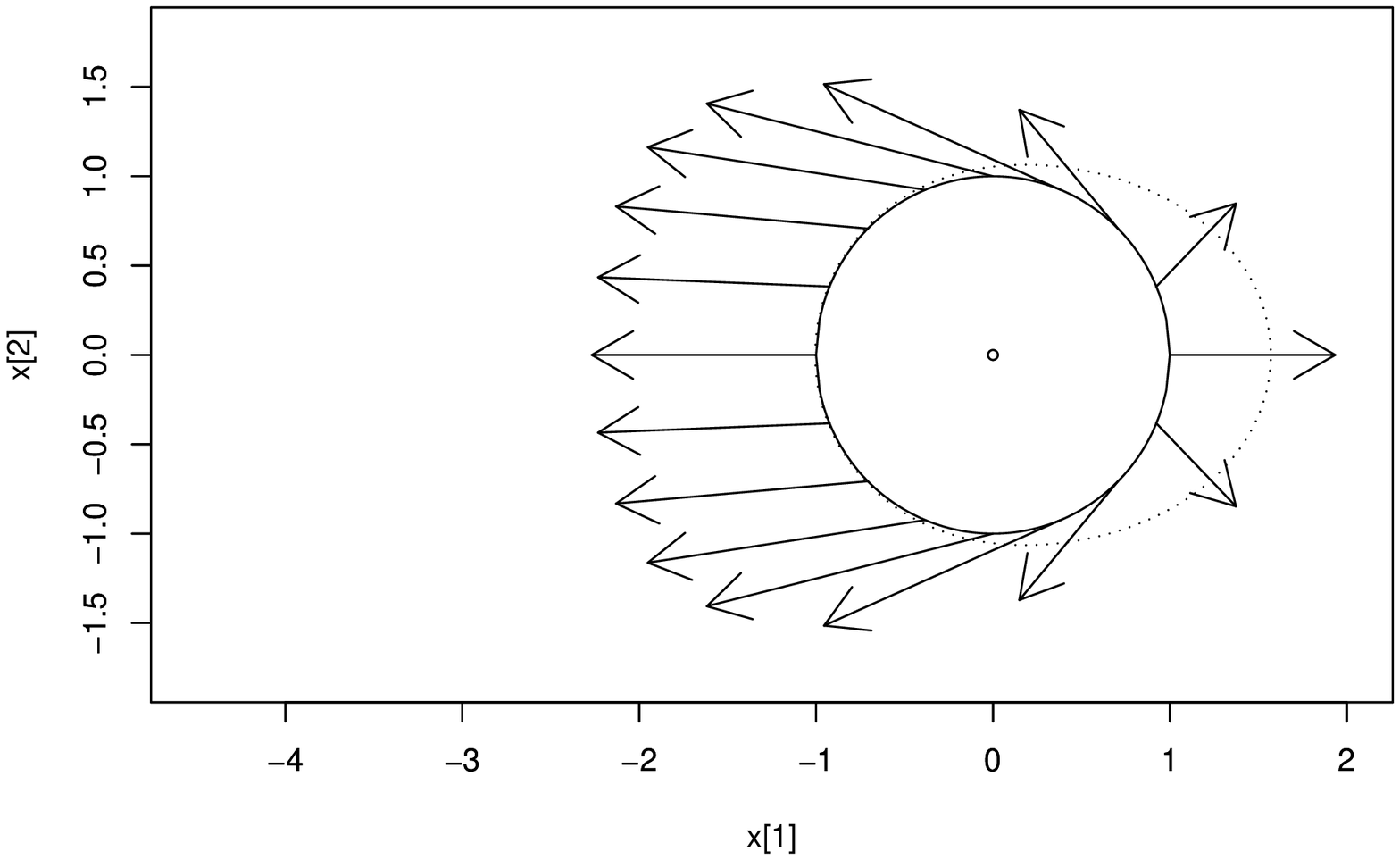}\\
\hspace{0.3cm} (c) \hspace{6.5cm} (d)\\
\caption[]{
\baselineskip 8.5mm
Influence functions (\ref{if_beta}) of the type 1 estimator for the $\mbox{vM}_2\{(2.37,0)'\}$ model with (a) $\beta=0.05,$ (b) $\beta=0.1,$ (c) $\beta=0.25$ and (d) $\beta=0.5$.
For convenience, the norms of the influence functions are divided by four.
The white dot denotes the origin and the dotted line represents the $\mbox{vM}_2\{(2.37,0)'\}$ density.}
\end{center}
\end{figure}

\newpage

\begin{figure}
\begin{center}
\includegraphics[width=6.8cm,height=4.25cm]{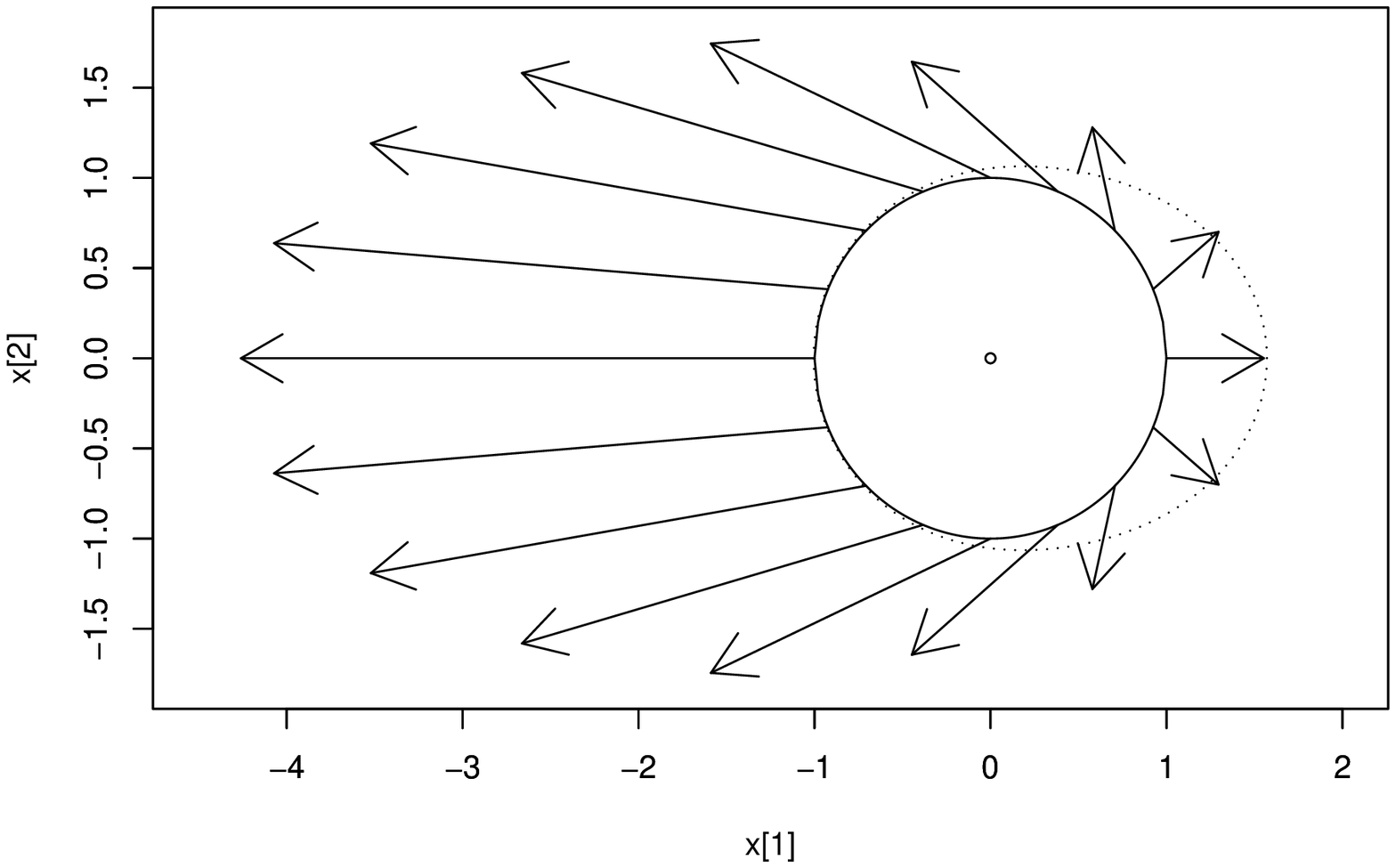}
\includegraphics[width=6.8cm,height=4.25cm]{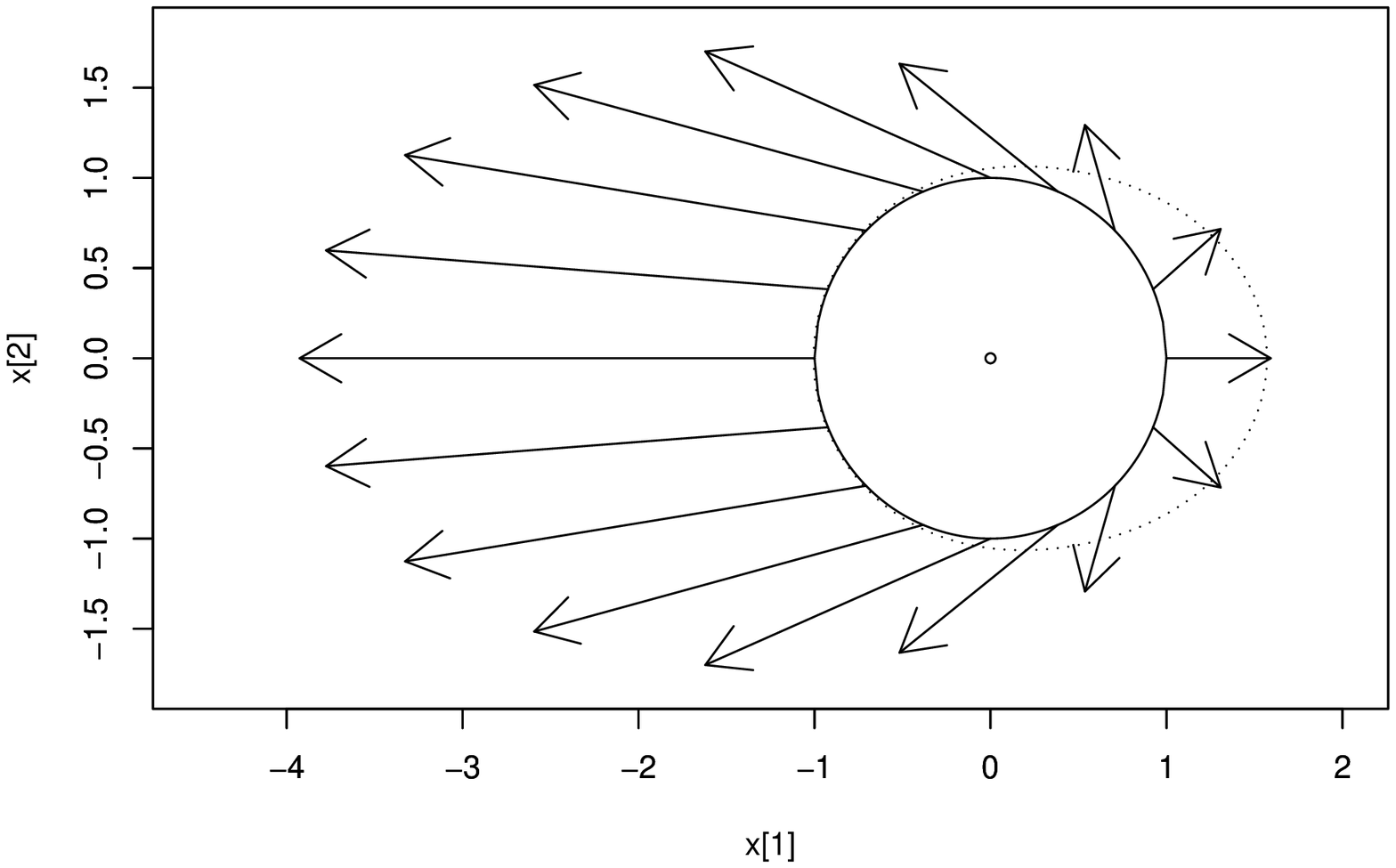}\\
\hspace{0.3cm} (a) \hspace{6.5cm} (b)\\
\includegraphics[width=6.8cm,height=4.25cm]{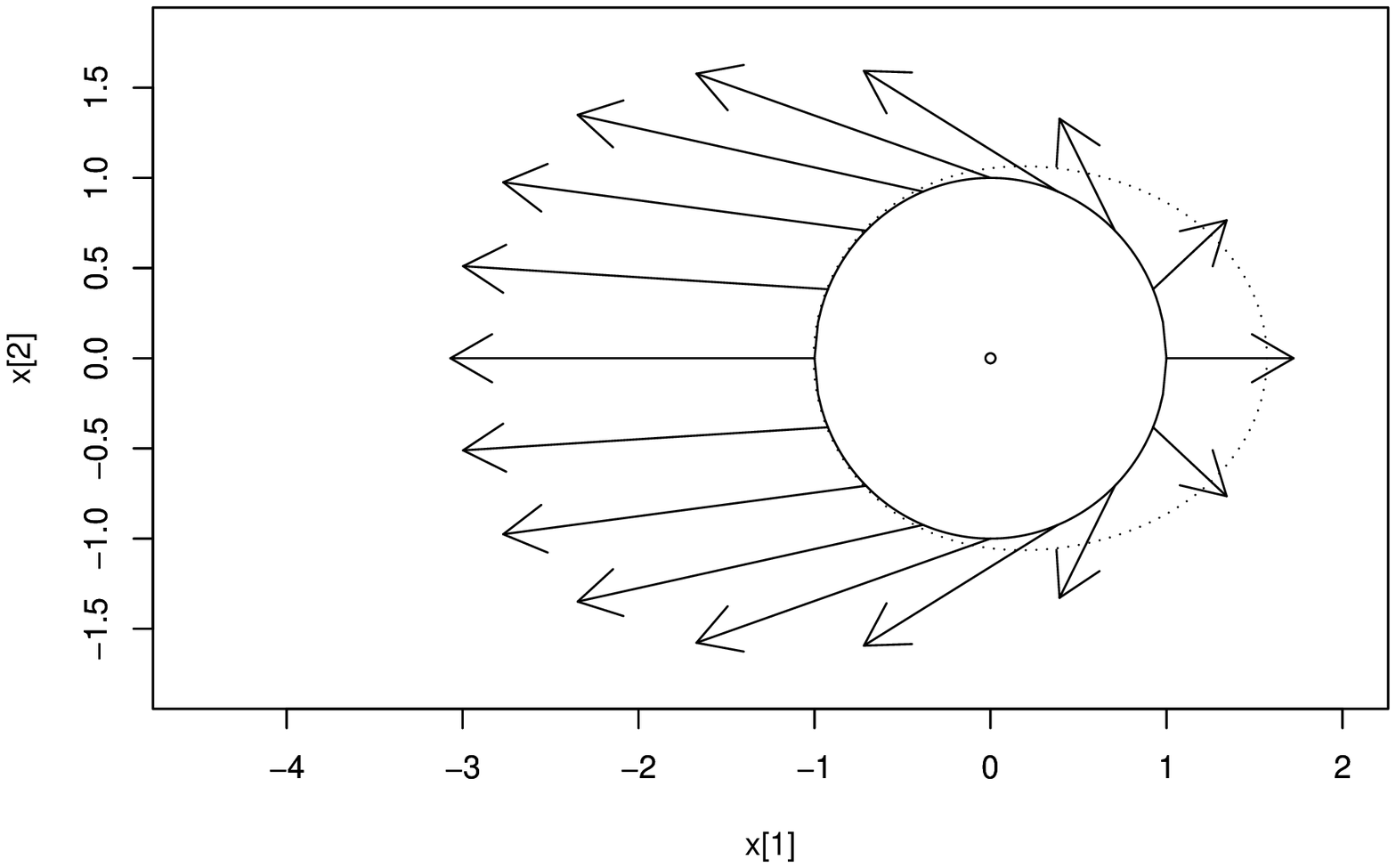}
\includegraphics[width=6.8cm,height=4.25cm]{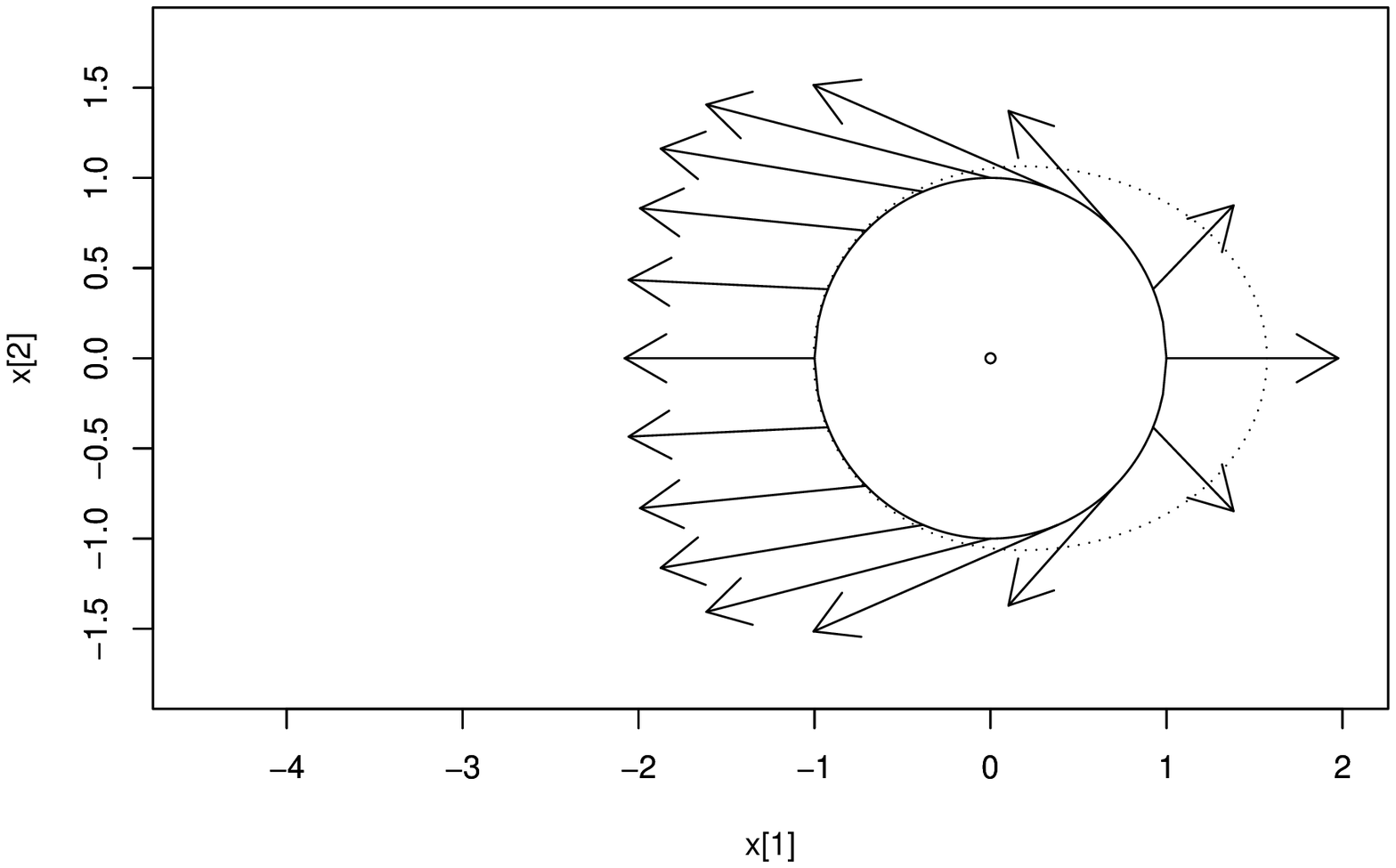}\\
\hspace{0.3cm} (c) \hspace{6.5cm} (d)\\
\caption[]{
\baselineskip 8.5mm
Influence functions (\ref{if_gamma}) of the type 0 estimator for the $\mbox{vM}_2\{(2.37,0)'\}$ model with (a) $\gamma=0.05,$ (b) $\gamma=0.1,$ (c) $\gamma=0.25$ and (d) $\gamma=0.5$.
For convenience, the norms of the influence functions are divided by four.
The white dot denotes the origin and the dotted line represents the $\mbox{vM}_2\{(2.37,0)'\}$ density.}
\end{center}
\end{figure}

\newpage
{\normalsize
\begin{figure}\vspace{-0.6cm}

\begin{center}
\hspace{1cm} \includegraphics[trim=0cm 0cm 0cm 0cm,clip,width=5.7cm,height=4.5cm]{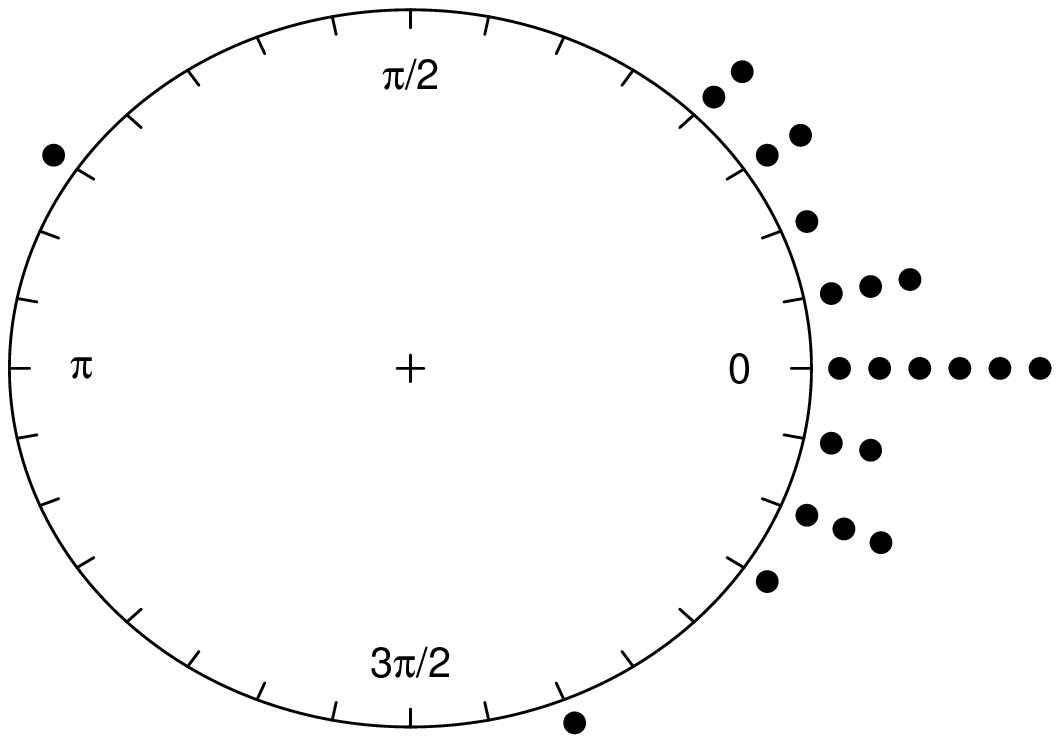}\\
(a) \\
\includegraphics[width=6cm,height=4.5cm]{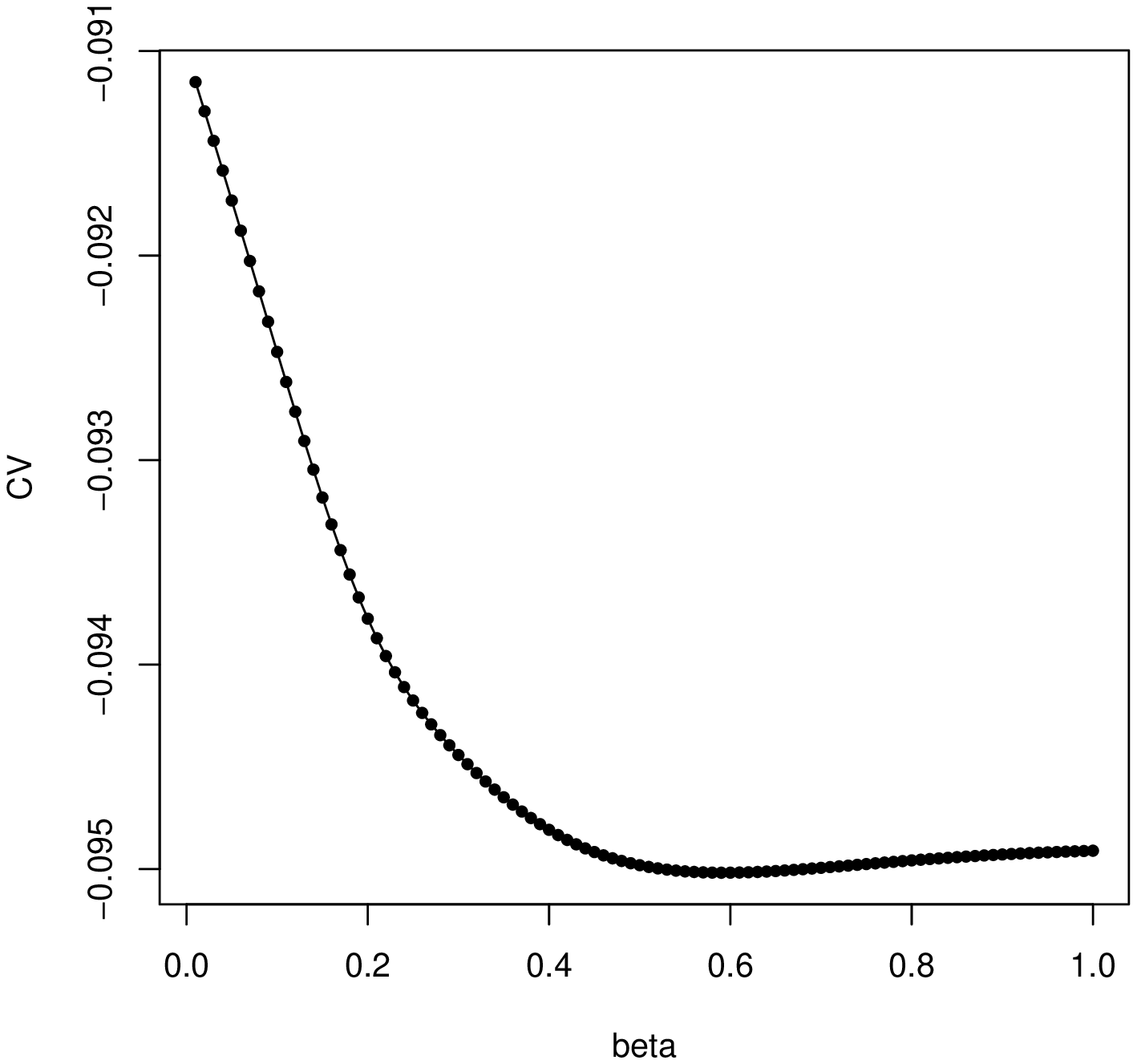}
\includegraphics[width=6cm,height=4.5cm]{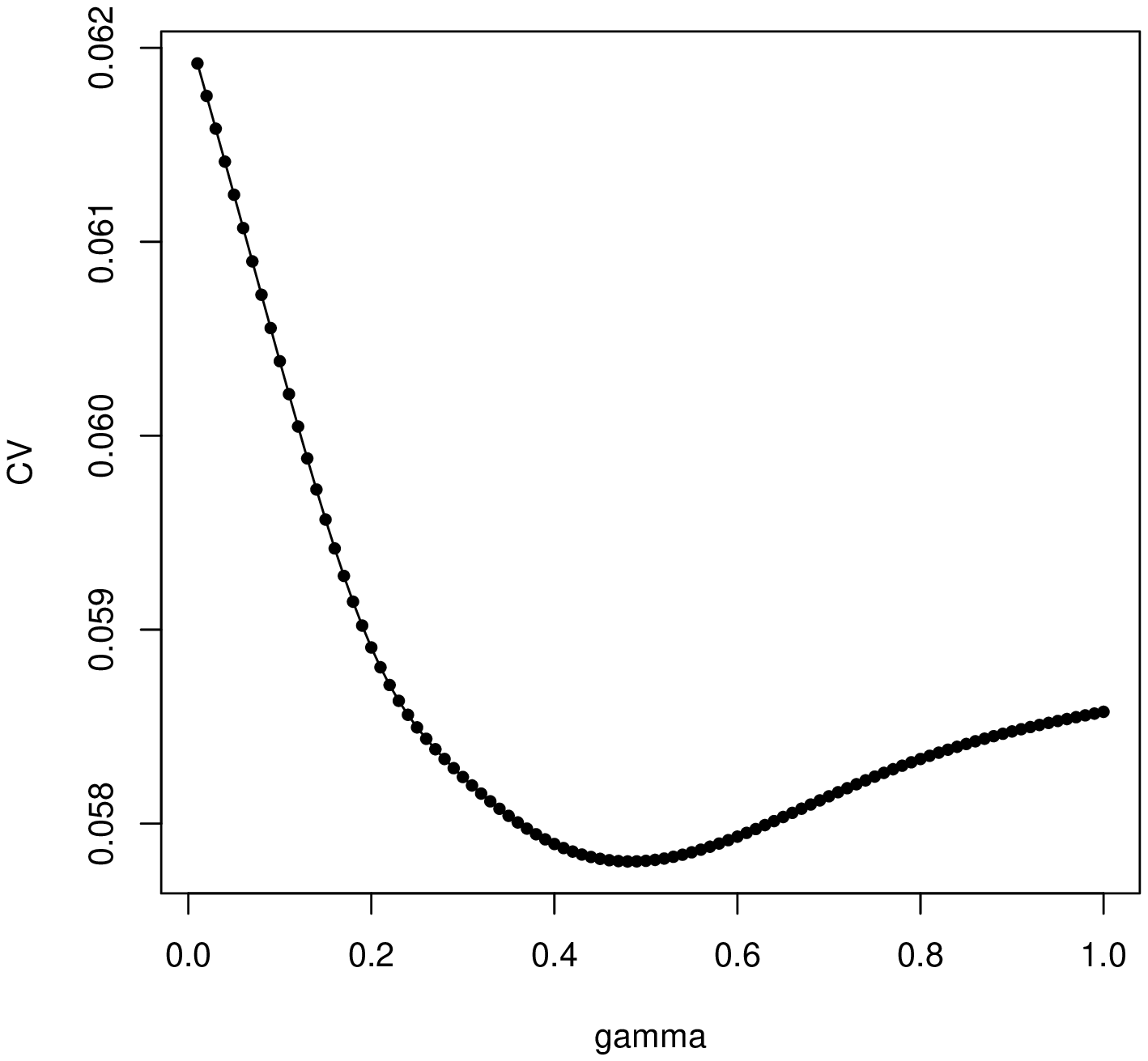}\\
\hspace{0.5cm} (b) \hspace{5.5cm} (c)\vspace{0.5cm}\\
\includegraphics[trim=0cm 0cm 0cm 0cm,clip,width=6cm,height=4.5cm]{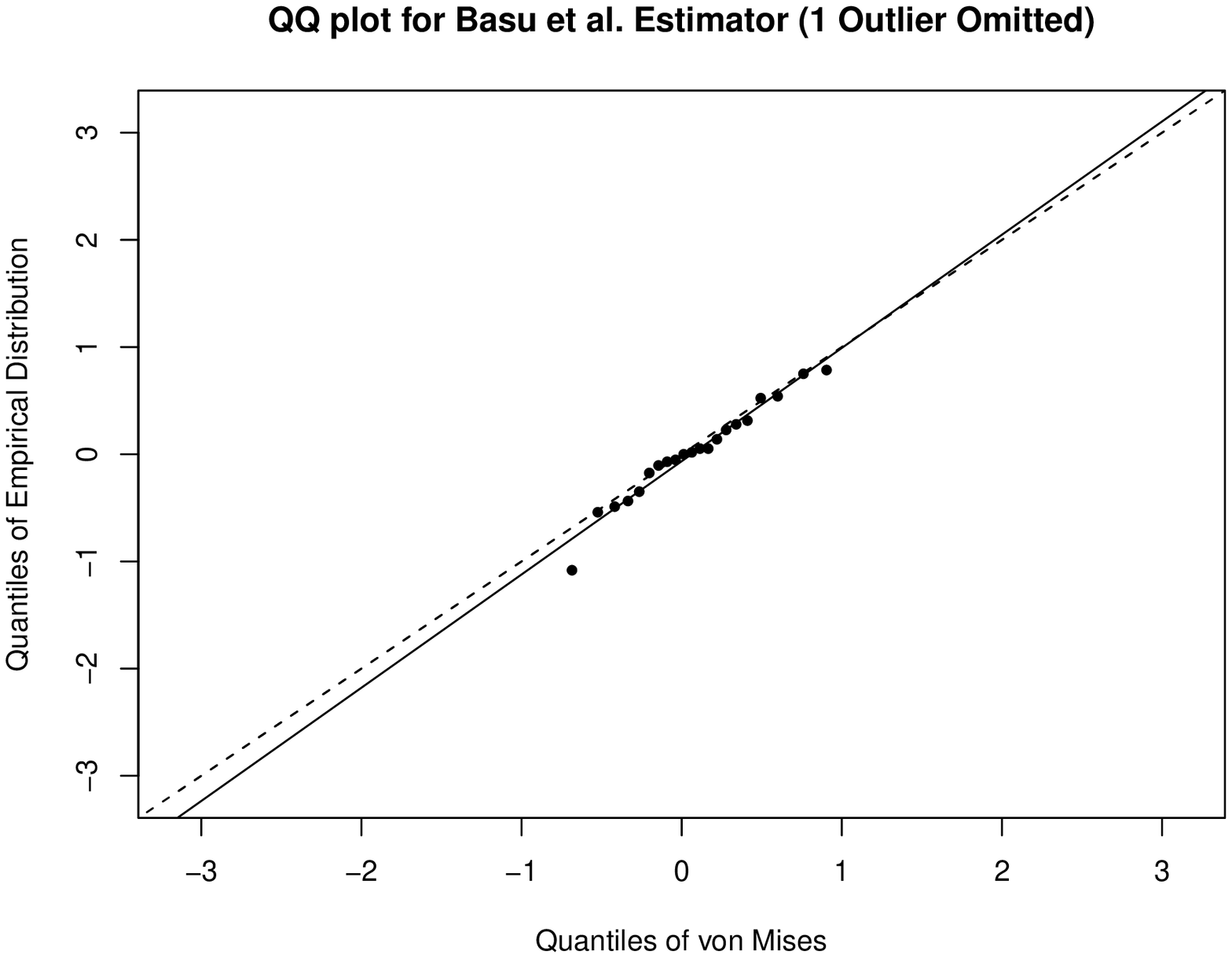}
\includegraphics[trim=0cm 0cm 0cm 0cm,clip,width=6cm,height=4.5cm]{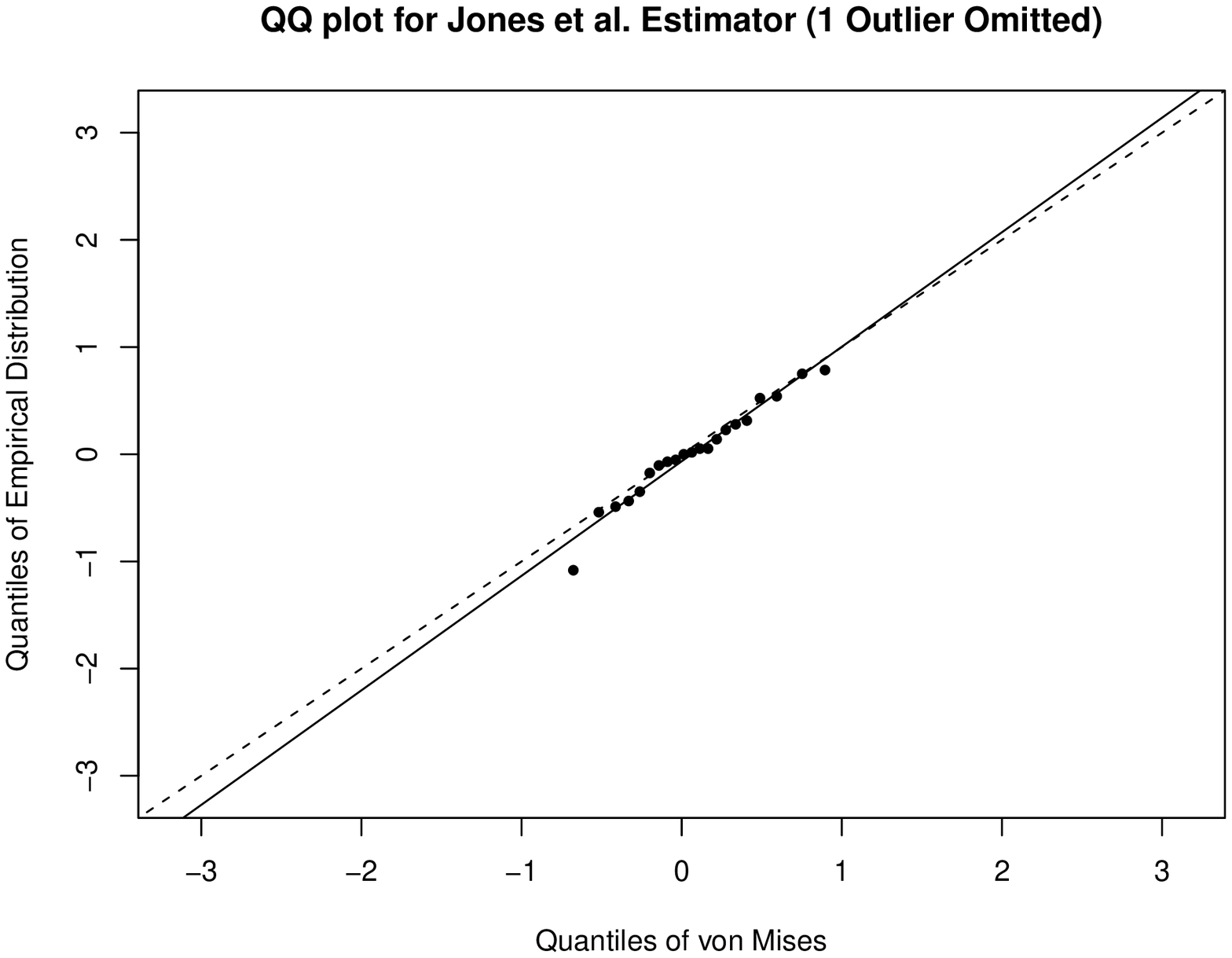}\\
\hspace{0.5cm} (d) \hspace{5.5cm} (e)
\caption[]{
\baselineskip 8.5mm
(a) Plot of measurements of resultant directions of 22 sea stars after 11 days of movement, plots of values of CV (\ref{cv}) for 100 selected values of tuning parameters between 0 and 1 for (b) type 1 and (c) type 0 estimators, and Q--Q plots for the data excluding one outlier for (d) type 1 estimator and (e) type 0 estimator where quantiles of the estimators ($x$-axis) and of the empirical distribution ($y$-axis) are plotted.}
\end{center}
\end{figure}
}

\end{document}